%% file: Main3_part_I.tex
\documentclass[journal]{IEEEtran}

\usepackage{mathptmx}       
\usepackage{helvet}         
\usepackage{courier}        
\usepackage{type1cm}        
\usepackage{makeidx}         
\usepackage{graphicx}        

\usepackage{multicol}        
\usepackage[bottom]{footmisc}
\usepackage{ifthen}
\usepackage{subfigure}
\usepackage[usenames,dvipsnames]{color}

\makeindex             

\usepackage{graphics} 
\usepackage{graphicx} 
\usepackage{epsfig} 
\usepackage{epstopdf}
\usepackage{mathptmx} 
\usepackage{times} 
\usepackage[cmex10]{amsmath}

\usepackage{mathtools}  
\usepackage{amssymb}  
\usepackage{amsthm}

\usepackage{amsfonts}

\usepackage{subfig}
\usepackage{verbatim}
\usepackage{comment}
\usepackage{algorithm}
\usepackage{algorithmic}
\usepackage{multirow}
\usepackage{cite}
\usepackage{stfloats}

\DeclareGraphicsExtensions{.pdf,.png,.jpg,.eps}

\def\ba{\begin{array}}
\def\ea{\end{array}}
\def\baa{\begin{align}}
\def\eaa{\end{align}}
\newcommand{\beq}{\begin{equation}}
\newcommand{\eeq}{\end{equation}}
\newcommand{\bq}{\begin{eqnarray}}
\newcommand{\eq}{\end{eqnarray}}
\newcommand{\bqn}{\begin{eqnarray*}}
\newcommand{\eqn}{\end{eqnarray*}}
\newcommand{\bee}{\begin{enumerate}}
\newcommand{\eee}{\end{enumerate}}
\newcommand{\bi}{\begin{itemize}}
\newcommand{\ei}{\end{itemize}}

\newcommand{\diag}{\mathrm{diag}}

\newboolean{showcomments}
\setboolean{showcomments}{true}
\newcommand{\wang}[1]{\ifthenelse{\boolean{showcomments}}
{ \textcolor[rgb]{1,0,1}{(ZW:  #1)}}{}}
\newcommand{\fliu}[1]{\ifthenelse{\boolean{showcomments}}
{ \textcolor{red}{(FL:  #1)}}{}}
\newcommand{\zhao}[1]{\ifthenelse{\boolean{showcomments}}
{ \textcolor{green}{(CZ:  #1)}}{}}
\newcommand{\slow}[1]{\ifthenelse{\boolean{showcomments}}
{ \textcolor{blue}{(SL:  #1)}}{}}

\theoremstyle{definition}
\newtheorem{theorem}{Theorem}
\newtheorem{lemma}[theorem]{Lemma}

\theoremstyle{definition}
\newtheorem{definition}{Definition}
\newtheorem{remark}{Remark}

\ifCLASSOPTIONcaptionsoff
\usepackage[nomarkers]{endfloat}
\let\MYoriglatexcaption\caption
\renewcommand{\caption}[2][\relax]{\MYoriglatexcaption[#2]{#2}}
\fi

\title{Distributed Frequency Control with Operational Constraints, Part I: Per-Node Power Balance}

        \author{Zhaojian~Wang, Feng~Liu, Steven~H.~Low,~\IEEEmembership{Fellow,~IEEE,}
                Changhong~Zhao,  and~Shengwei~Mei~\IEEEmembership{Fellow,~IEEE}
                \thanks{This work was supported  by the National Natural Science Foundation
                        of China ( No. 51677100, No. 51377092,  No. 51621065), Foundation of Chinese Scholarship Council (CSC No. 201506215034), 
the US National Science Foundation through awards EPCN 1619352, CCF 1637598, 
CNS 1545096, ARPA-E award DE-AR0000699, and Skoltech through Collaboration
                        Agreement 1075-MRA.     }       
                \thanks{Z. Wang, F. Liu and S. Mei are with the Department
                        of Electrical Engineering, Tsinghua University, Beijing,
                        China, 100084 e-mail: (lfeng@tsinghua.edu.cn).}
                \thanks{S. H. Low and C. Zhao  are with the Department
                        of Electrical Engineering, California Institute of Technology, Pasadena, CA, USA, 91105 e-mail:(slow@caltech.edu)}
        }

\begin{document}

                \maketitle

\begin{abstract}                                                
    This paper addresses the distributed optimal frequency control of multi-area power system with operational constraints, including the regulation capacity of individual control area and the power limits on tie-lines. Both generators and controllable loads are utilized to recover nominal frequencies while minimizing regulation cost. We study two control modes: the per-node balance mode and the network balance mode. In Part I of the paper, we only consider the per-node balance case, where we derive a \emph {completely decentralized} strategy without the need for communication between control areas. It can adapt to unknown load disturbance. The tie-line powers are restored after load disturbance, while the regulation capacity constraints are satisfied both at equilibrium and during transient.  We show that the closed-loop systems with the proposed control strategies carry out primal-dual updates for solving the associated centralized frequency optimization problems. We further prove  the closed-loop systems are asymptotically stable and converge to the unique optimal solution of the centralized frequency optimization problems and their duals. Finally, we present simulation results to demonstrate the effectiveness of our design. In Part II of the paper, we address the network power balance case, where transmission congestions are managed continuously. 
%
%
\end{abstract}
                

                \IEEEpeerreviewmaketitle

\input{sec-intro_part_I}

\input{sec-formulation_part_I}

\input{sec-per-node-v6}
\input{sec-sim_part_I}

\input{sec-conclusion_part_I}

\bibliographystyle{IEEEtran}
\bibliography{mybib,PowerRef-201202}

\input{sec-appendix_part_I}

\end{document}

%% file: sec-intro_part_I.tex
\section{Introduction}

In a modern large-scale power system, multiple regional grids are usually interconnected for improving operation reliability and economic efficiency \cite{Min:Total, Ahmadi:Probabilistic}. In each control area, power generation and controllable load can be utilized to eliminate power imbalance and maintain frequency stability in real time. Generally, frequency control is a paid service, and hence control areas always try to minimize their control cost. As different control areas may belong to different utilities and global information may not be accessible due to privacy and operational considerations, a distributed strategy is desirable. Roughly speaking, there are two possible modes of operation. In the first mode, each node area balances its own supply and demand after a disturbance. Then the power flow on each tie line should be regulated to its scheduled value, i.e., the deviation of power flows on the tie lines are eliminated in equilibrium. In the second mode, all nodes cooperate to rebalance power
over the entire network after a disturbance. The power flows on the tie lines may deviate from their scheduled values but must satisfy line limits in equilibrium. We refer to the first case as \emph{per-node (power) balance} and the second \emph{network (power) balance}. Here we focus on the first case, while the second case will be addressed in Part II of the paper. We design a decentralized optimal frequency controller for restoring frequency and tie-lime power under operational constraints, including regulation capacity constraints.

Different distributed strategies have been developed in the  literature for frequency control.
They can roughly be divided into two categories in terms of different types of regulation resources: the automatic generation control (AGC) e.g. \cite{Ibraheem:Recent, Variani:Distributed, Stegink:aunifying, LiZhaoChen2016, Stegink:aport, Wang:Distributed, Wang:Decentralized} and the load-side frequency control e.g.\cite{Ilic:Modeling, Zribi:Adaptive, Changhong:Design, Changhong:Decentralized, Mallada-2017-OLC-TAC, Kasis:Primary1, Devane:Primary2}. The former  focuses on generation regulation. 
For example in 
\cite{Variani:Distributed} a flatness-based control combining trajectory generation and trajectory tracking is proposed for AGC in multi-area power system.  {In \cite{Stegink:aunifying, Stegink:aport}, the closed-loop system composed of power system dynamics and controller dynamics is formulated as a  port-Hamiltonian system, and its stability is proved. } In \cite{LiZhaoChen2016}, generators are driven by AGC to restore  frequencies. Correspondence between the (partial) primal-dual gradient algorithm for solving the associated optimization problem and the  frequency control dynamics of the physical system is established. The resulting decomposition enables  the system design of a fully distributed optimal frequency control.

For the load-side frequency control, 
load frequency dynamics are formulated similarly to the generator model in \cite{Ilic:Modeling}, leading to a distributed frequency control for both generation and controllable loads. 
A distributed adaptive control is presented in \cite{Zribi:Adaptive} to guarantee acceptable frequency deviation from the nominal value.  
In \cite{Changhong:Design, Changhong:Decentralized}, an optimal load control (OLC) problem is formulated and a ubiquitous primary load-side control is derived as a partial primal-dual gradient algorithm for solving the OLC problem.  It is decentralized, but does not restore the nominal frequency. 
This design approach is extended in \cite{Mallada-2017-OLC-TAC} to secondary control
that restores  nominal frequency and scheduled inter-area flows as well as enforcing
line limits in equilibrium.  {It is further extended to more general  models in \cite{Kasis:Primary1,Devane:Primary2}, where passivity condition guaranteeing stability is proposed for each local bus and the conservativeness is reduced greatly.}

In terms of methodology, there are mainly three types of distributed  frequency control: the droop based approach e.g. \cite{Maknouninejad:Optimal,Nasirian:Distributed}, the consensus based approach e.g. \cite{Olfati:Consensus,Binetti:Distributed, Xin:Cooperative} and the primal-dual decomposition based approach e.g.\cite{Changhong:Design,Changhong:Decentralized, Mallada-2017-OLC-TAC, LiZhaoChen2016}. 
In the primal-dual decomposition approach, control goals such as rebalancing power after a disturbance and
restoring nominal frequency and scheduled inter-area flows are formalized as a global optimization problem.
The feedback control laws are designed so that the equilibrium of the closed-loop system solves the
optimization problem and hence achieves the control objectives, in equilibrium. 
Moreover the closed-loop system is designed to be an asymptotically stable primal-dual algorithm 
for solving an associated optimization problem. This is the same approach taken in \cite{jokic:real} where the real-time control is through nodal proces.

In all the primal-dual algorithms proposed in the literature, even though input constraints are usually
enforced, constraints on states, such as power injections on buses, are enforced
only in the steady-state.  In practice, however, a control area always maintains regulation capacity 
bounds that constrain power generations and controllable loads within  available ranges, not only at equilibrium but also during transient. 
In this paper, we design input-saturation controllers that maintain these capacity constraints 
during transient as well. 
We show that these controllers still carry out primal-dual updates of the associated 
optimization problem. 
 
Specifically we study an approach to rebalancing power after a disturbance.
In the per-node power balance case, we require the disturbance in each control area
be balanced by generations and controllable loads in that area.  Then
we construct a completely decentralized control to recover nominal frequencies 
and tie-line power flows.
The  regulation capacity constraints are also enforced during transient.
We show that the controller together with the physical dynamics serve as  primal-dual updates with saturation for solving the optimization problem. Then we prove the optimality of our control by exploiting the equivalence between the equilibrium of the closed-loop frequency control system and  the optimal solution of the optimization problem. 
We also show that the optimal solution of the primal-dual problem and equilibrium point of the closed-loop
systems are both unique. Furthermore we prove the stability of closed-loop system by combining projection technique with LaSalle's invariance principle, mitigating the impact of nonsmooth dynamics created by the imposed transient constraints.  The salient features of our control are: 
\bee
\item \emph{Control goals:} the controller restores the nominal frequency and tie-line powers after unknown disturbance while minimizing the regulation costs; 

\item \emph{Constraints:} the regulation capacity constraints are always enforced even during transient; 


\item \emph{Communication:} it is completely decentralized without the need for communication among neighboring areas ;

%
%

\item \emph{Measurement:} the controls are adaptive to unknown load disturbances automatically without load measurement.
\eee

The rest of this paper is organized as follows. In Section II, we introduce our network model. Section III formulates the optimal frequency control problem,  presents our controller and its relationship with the primal-dual update, and proves the optimality, uniqueness, and stability of the closed-loop equilibrium point. 
We confirm the performance of  controllers via simulations on a detailed power system model in Section IV. Section V concludes the paper.

%% file: sec-formulation_part_I.tex
\section{Network model}


A large power network is usually composed of multiple control areas
each with its own generators and loads.  These control areas are interconnected
with each other through tie lines.  For simplicity, here we treat each control area
as a node with an aggregate power generation, an aggregate controllable load
and an aggregate uncontrollable load.\footnote{In our study, each of the nodes can be regarded as a control area including controllable generation and load. All controllable generations in the same control area are aggregated into one equivalent generator, while all controllable loads are aggregated into one controllable load. 
}
Then the power network is model by a graph ${G}:=(N, E)$ 
where  $N=\{0,1,2,...n\}$ is the set of nodes (control areas) and
$E\subseteq N\times N$ is the set of edges (tie lines).  If a pair of
nodes $i$ and $j$ are connected by a tie line
directly, we denote the tie line by $(i,j)\in E$.
Let $m:= |E|$ denote the number of tie lines.
We treat $G$ as directed with an arbitrary orientation and we use $(i,j)\in E$
or $i\rightarrow j$ interchangeably to denote a directed edge from $i$ to $j$.
It should be clear from the context which is the case.
Without loss of generality, we assume the graph is connected
and node $0$ is a reference node.

For each node $j\in N$, let $\theta_j(t)$ denote the rotor angle at node $j$ at time
$t$ and $\omega_j(t)$ the  frequency.
Let $P_j^g(t)$ denote the (aggregate) generation at node
$j$ at time $t$ and $u^g_j(t)$ its generation control command.
Let $P^l_j(t)$ denote the (aggregate) controllable load and $u^l_j(t)$ its
load control command.  Let $p_j$ denote the (aggregate) uncontrollable load. 

We adopt a second-order linearized model to describe the frequency dynamics
of each node, and two first-order inertia equations to describe the dynamics of
power generation regulation and load regulation at each node.  We assume
the tie lines are lossless and adopt the DC power flow model.
Then for each node $j\in N$,
\begin{subequations}
        \begin{align}
        \dot \theta_j & =  \omega_j(t)
        \label{eq:model.1a}
        \\
        M_j \dot \omega_j & =   P^g_j(t) - P^{l}_j(t) - p_j -D_j \omega_j(t)
        \nonumber
        \\
        &  + \sum_{i: i\rightarrow j} \! B_{ij} (\theta_i(t) - \theta_j(t))
                -  \sum_{k: j\rightarrow k} \! B_{jk}(\theta_j(t) - \theta_k(t))
        \label{eq:model.1b}
        \\
        T^g_j \dot P^g_j & =  - P^g_j(t) + u^g_j(t) - {\omega_j(t)}/{R_j}
        \label{eq:model.1c}
        \\
        T^l_j \dot P^{l}_j & =  - P^{l}_j(t) + u^l_j(t)
        \label{eq:model.1d}
        \end{align}
        where $D_j>0$ are damping constants, $R_j>0$ are droop parameters,
        and $B_{jk}>0$ are line parameters that depend on the reactance of the line $(j,k)$.
        Let  $x := (\theta, \omega, P^g, P^l)$ denote the state of the network
        and $u := (u^g, u^l)$ denote the control.\footnote{Given
        	a collection of $x_i$ for $i$ in a certain set $A$, $x$ denotes the column vector
        	$x := (x_i, i\in A)$ of a proper dimension with $x_i$ as its components.}
        \label{eq:model.1}
\end{subequations}

Our goal is to design feedback control laws for the generation command
$u^g(x(t))$ and load control $u^l(x(t))$.
The operational constraints are:
\begin{subequations}
        \bq
                \underline{P}^g_j & \leq \ P^g_j(t) \ \leq \overline{P}^g_j, \quad j\in N
        \label{eq:OpConstraints.1a}\\
                \underline{P}^l_j & \leq \ P^l_j(t) \ \leq \overline{P}^l_j, \quad j\in N
        \label{eq:OpConstraints.1b}
        \eq
        
        Differing from the literature, here \eqref{eq:OpConstraints.1a}  and \eqref{eq:OpConstraints.1b} are  \emph{hard limits} on the regulation capacities of generation and controllable load at each node, which should not be violated at any time even during transient.
     \label{eq:OpConstraints.1}
\end{subequations}
Hence we will design controllers so that these constraints are satisfied not only at equilibrium, but also during transient.

We assume that the system operates in a steady state initially, i.e., the generation
and the load are balanced and the frequency is at its nominal value.   All variables
represent \emph{deviations} from their nominal or scheduled values so that,
e.g., $\omega_j(t)=0$ means the frequency is at its nominal value.

As the generation $P^g_j$ and load $P^l_j$ in each area can increase or decrease, 
and a line flow $P_{ij}$ can in either direction, we make the following assumption.
\bi
	\item[\textbf{A1:}] 
	\bee
	\item $\underline{P}^g_j < 0 < \overline{P}^g_j$ and $\underline{P}^l_j <0 < \overline{P}^l_j$
	 for $\forall j\in N$.
	\item $\theta_0(t):=0$ for all $t\geq 0$.
	\eee
\ei  
The assumption $\theta_0 \equiv 0$ amounts to using $\theta_0(t)$ as reference
angles.  It is made merely for notational convenience: as we will see, the equilibrium point will be
unique with this assumption (or unique up to reference angles without this assumption).

%% file: sec-per-node-v6.tex
\section{Per-Node Power Balance}

First we consider the per-node power balance case, modeled by the requirement:
\bq
P^g_j & = &  P^l_j + p_j, \qquad j\in N
\label{eq:balance.node}
\eq
%

\subsection{Control goals}
The control goals are formalized as an optimization problem:
\begin{subequations}
        \bq
      \text{PBO:~}  \min & \!\!\!\!\!\! & \frac{1}{2} \sum_j \alpha_j \left(P^g_j\right)^2 
        + \frac{1}{2} \sum_j \beta_j \left(P^l_j\right)^2
        + \frac{1}{2} \sum_j D_j \omega_j^2
        \nonumber
        \\
        \label{eq:opt.1d}
        \\
        \text{over} & \!\!\!\!\!\! & x := ({\theta}, \omega, P^g, P^l) \text{ and }
        u := (u^g, u^l)
        \nonumber
        \\ 
        \text{s. t.}  
        & \!\!\!\!\!\! & \eqref{eq:OpConstraints.1},  \eqref{eq:balance.node}
        \nonumber
        \\
        & \!\!\!\!\!\! & 
        P^g_j = P^l_j  + p_j +  U_j(\theta, \omega_j), \quad j\in N
        \label{eq:opt.1a}
        \\
        & \!\!\!\!\!\! & P^g_j \ = \ u^g_j, \quad j\in N
        \label{eq:opt.1b}
        \\
        & \!\!\!\!\!\! & P^l_j \ = \ u^l_j, \quad j\in N
        \label{eq:opt.1c}
        \eq
        \label{eq:opt.1}
\end{subequations}
where $\alpha_j>0$, $\beta_j>0$ are constant weights and 
\bqn
U_j(\theta, \omega_j) & :=  & D_j \omega_j  
- \sum\nolimits_{i: i\rightarrow j} B_{ij} \theta_{ij}
+ \sum\nolimits_{k: j\rightarrow k} B_{jk} \theta_{jk}
\eqn
Here we have abused notation and use
$\theta_{ij} := \theta_i - \theta_j$.
In vector form we have
\bq
U(\theta, \omega) & := & D\omega + CBC^T \theta
\label{eq:defE.1}
\eq
where $D := \diag(D_i, i\in N)$, $B := \diag(B_{ij}, (i,j)\in E)$, 
$C$ is the $(n+1)\times m$ incidence matrix.

We comment on the optimization problem \eqref{eq:opt.1}.
\begin{remark}[Control goals]
        \bee
        \item 
        Since the variables are deviations from their nominal values, the
        parameters $(\alpha_j, \beta_j)$ in the objective function \eqref{eq:opt.1d}
        are not electricity costs.  
        Minimizing the objective aims to track 
        generation and consumption that have been scheduled at a slower timescale,
        e.g., to optimize economic efficiency or user utility.
        The parameters $(\alpha_j, \beta_j, D_j)$  weigh the relative costs of 
        deviating from scheduled generation and load, and the nominal frequency. In the next subsection we will show that, for every optimal solution, the corresponding frequency deviation must be zero, provided a feasible solution exists.
        
        \item
        For the definition of \eqref{eq:opt.1}, the regulation capacity limits 
        \eqref{eq:OpConstraints.1} apply only at optimality. As we will see below, our controller, however, enforces \eqref{eq:OpConstraints.1} even during transient.
                
        \item The per-node balance requirement \eqref{eq:balance.node} 
        and the constraint \eqref{eq:opt.1a} imply $U(\theta, \omega) = 0$
        at any feasible $x$.   This will drive the power flow on 
        \emph{every} tie line to its scheduled value, i.e.,
        $P_{ij}^*=0$ in equilibrium (see Theorem \ref{thm:1} below),
        even though this is not included in \eqref{eq:opt.1} as a constraint.
        
        \item The constraints \eqref{eq:opt.1b}\eqref{eq:opt.1c} require that, at optimality,
        the power injection $P^g_j$ and controllable load $P^l_j$ are equal to their
        control commands $u^g_j$ and $u^l_j$ respectively.
        \eee
\end{remark}

In the rest of the paper we make one of the following assumptions
(recall that $(\underline P^g, \underline P^l) < 0 < (\overline P^g, \overline P^l)$
under  A1):
\bi
\item[\textbf{A2}:] The PBO problem \eqref{eq:opt.1} is feasible, i.e.,  
	\bqn
	\underline{P}^g_j - \overline{P}^l_j & \leq \ p_j \ \leq & \overline{P}^g_j - \underline{P}^l_j,
	\qquad \forall j\in N
	\eqn
	Moreover \eqref{eq:opt.1} has a finite optimal solution.
\ei  
Feasibility of \eqref{eq:opt.1} is equivalent to the inequalities in A2 because
the per-node balance constraint \eqref{eq:balance.node} requires 
$p = P^g - P^l$ in equilibrium.
In what follows below, we sometimes strengthen the inequalities in A2 to strict
inequalities.  Strict inequalities mean that each area has a certain power margin. 
If there is no margin
the system may have no feasible solution after a small load disturbance. 
For example, if $\overline{P}^g_j - \underline{P}^l_j = p_j$ for any area $j$, then 
any feasible solution must have ${P}^{g}_j=\overline{P}^g_j$ and 
${P}^{l}_j=\underline{P}^l_j$, i.e., there is no more regulation capacity in area $j$
so that if the load $p_j$ further increases, then frequency will drop and cannot be restored.

\subsection{Decentralized controller}
\label{subsec:controllers}

Our  control laws for $u^g$ and $u^l$ are: for each node $j\in N$,
\begin{subequations}
	\bq
	\dot \lambda_j & = &  \gamma^{\lambda}_j \left( P^g_j(t) -  P^l_j(t) - p_j\right)
	\label{eq:control.1c}
	\\
	u^g_j(t) & = & \left[ 
	P^g_j(t) - \gamma^g_j \left( \alpha_j P^g_j(t) + \omega_j(t) + \lambda_j(t) \right)
	\right]_{\underline P^g_j}^{\overline P^g_j}
	\nonumber \\
	&&+ {\omega_j(t)}/{R_j}
	\label{eq:control.1a}
	\\
	u^l_j(t) & = & \left[ 
	P^l_j(t) - \gamma^l_j \left( \beta_j P^l_j(t) - \omega_j(t) - \lambda_j(t) \right)
	\right]_{\underline P^l_j}^{\overline P^l_j}
	\label{eq:control.1b}
	\eq
	where $\gamma^g_j, \gamma^l_j, \gamma^{\lambda}_j$ are positive constants.
	For any $x_i, a_i, b_i \in \mathbb R$ with $a_i\leq b_i$, 
	$[x_i]_{a_i}^{b_i} := \min \{ b_i, \max \{ a_i, x_i \} \}$. 
	\label{eq:control.1}
\end{subequations}
For vectors $x, a, b$, $[x]_a^b$ is defined accordingly componentwise.

The controller \eqref{eq:control.1} has a simple proportional-integral 
(PI) structure with saturation. 
It is \emph{completely decentralized} where each node $j$ updates its
internal state $\lambda_j(t)$ in \eqref{eq:control.1c} based only on the generation 
$P^g_j(t)$, the controllable load $P^l_j(t)$ and the uncontrolled load $p_j$ 
that are all local at $j$ (within a control area).   
The control inputs $u^g_j(t)$ and $u^l_j(t)$
in \eqref{eq:control.1a} and \eqref{eq:control.1b}
are then static functions of the local state $(P^g_j(t), P^l_j(t), \omega_j(t))$ and the
internal state $\lambda_j(t)$.  Therefore, no communication is required even between nodes.  

We often write $u^g_j$ and $u^l_j$ as functions of 
$(P^g_j, P^l_j, \omega_j, \lambda_j)$:
\begin{subequations}
	\bq
	u^g_j(t) & := & u^g_j \left( P^g_j(t), \omega_j(t), \lambda_j(t) \right)
	\label{eq:control.1a'}
	\\
	u^l_j(t) & := & u^l_j \left( P^l_j(t), \omega_j(t), \lambda_j(t) \right)
	\label{eq:control.1b'}
	\eq
	 for $j\in N$, where these functions are given by the right-hand side of \eqref{eq:control.1a}
	\eqref{eq:control.1b}.
	We now comment on measurements required to implement the control
	\eqref{eq:control.1}. 
	\label{eq:control.1'}
\end{subequations}
\begin{remark}[Implementation]
	The variable $\lambda_j(t)$ in \eqref{eq:control.1c} is a cyber quantity that is
	computed at each node $j$ based on $(P^g_j(t), P^l_j(t), p_j)$ locally at $j$ 
	(within a control area). These quantities can in principle be measured at $j$.  
	We would however like to avoid measuring the uncontrolled load change $p_j$ 
	for ease of implementation.
	To this end let $\Delta P_j(t) := P^g_j(t) - P^l_j(t) - p_j$, $j\in N$, denote the 
	surplus generation at node $j$.   We then have from \eqref{eq:model.1b} and 
	\eqref{eq:defE.1} that $\Delta P_j(t) = M_j \dot\omega_j + U_j(\theta, \omega_j(t))$.
	Since $\dot\lambda_j = \gamma^{\lambda}_j \Delta P_j(t)$, \eqref{eq:control.1c} becomes: 
	\bqn
	\dot\lambda_j & \!\!\!\! = \!\!\!\! & 
	\gamma^{\lambda}_j M_j \dot \omega_j \, + \, \gamma^{\lambda}_j D_j\omega_j(t)   
	\, - \, \gamma^{\lambda}_j \!\! \bigg(
	\sum_{i: i\rightarrow j} \! P_{ij}(t)
	-  \!\! \sum_{k: j\rightarrow k} \! P_{jk}(t)\!\! \bigg)
	\eqn
	where $P_{ij}(t) := B_{ij}(\theta_i(t)-\theta_j(t))$ are the tie-line flows from nodes $i$ to
	$j$ according to the DC power flow model. 
	Hence, to update the internal state $\lambda_j(t)$,
	we only need to measure the local frequency deviation $\omega_j(t)$, its derivative
	$\dot\omega_j(t)$ and the tie-line flows $P_{ij}(t)$ incident on node $j$,
	and not the uncontrolled load $p_j$ in area $j$.  
	An important advantage is that the controller naturally adapts to unknown load 
	changes $p_j$.  This feature will be illustrated in case studies. 
	
	The control inputs $u^g_j(t)$ and $u^l_j(t)$ in \eqref{eq:control.1'}
	can then be implemented 
	using measurements of the local generation $P^g_j(t)$, controlled load 
	$P^l_j(t)$, frequency deviation $\omega_j(t)$ and tie line powers $P_{ij}(t), P_{jk}(t)$. 
\end{remark}

\subsection{Design rationale}
\label{subsec:design}

The controller design \eqref{eq:control.1} is motivated by an approximate primal-dual algorithm for
\eqref{eq:opt.1}.   We first review the form of a standard primal-dual algorithm and
then explain that {the closed-loop dynamics
	\eqref{eq:model.1}\eqref{eq:control.1} carry out an approximate version for 
	\eqref{eq:opt.1} in real time over the closed-loop system.} 
	
\vspace{0.1in}
\noindent
\textbf{Primal-dual algorithms.}
Consider a general  constrained convex optimization:
\bqn
\min_{x\in X} \ \ f(x) & s. t. & g(x) = 0
\eqn
where $f:\mathbb R^n\rightarrow \mathbb R$, $g:\mathbb R^n\rightarrow \mathbb R^k$,
and $X\subseteq \mathbb R^n$ is closed and convex.   
Let $\rho\in\mathbb R^k$ be the Lagrange multiplier associated
with the equality constraint $g(x)=0$.  Define the Lagrangian $L(x; \rho) := f(x) + \rho^T g(x)$.
A standard primal-dual algorithm takes the form:
\begin{subequations}
\bq
x(t+1) & := & \text{Proj}_X \left( x(t) \ - \ \Gamma^x\, \nabla_x L(x(t); \rho(t)) \right)
\label{eq:pma.1a}
\\
\rho(t+1) & := & \rho(t) \ + \ \Gamma^\rho\, \nabla_\rho L(x(t); \rho(t))
\label{eq:pma.1b}
\eq
where Proj$_X(a)$ projects $a\in\mathbb R^n$ to the closest point in $X$
under the Euclidean norm, 
the gain matrices $\Gamma^x, \Gamma^\rho$ are (strictly) positive definite.
\label{eq:pma.1}
\end{subequations}
Hence the iterates $(x(t), \rho(t))$ stays in the set $X\times\mathbb R^k$ for all $t$
and, under appropriate assumptions, converges to a primal-dual optimal point.

In contrast a standard dual algorithm takes the form:
\begin{subequations}
\bq
\rho(t+1) & := & \rho(t) \ + \ \Gamma^\rho\, \nabla_\rho L(x(t); \rho(t))
\label{eq:da.1a}
\\
x(t) & := & \min_{x\in X}\, L(x; \rho(t))
\label{eq:da.1b}
\eq
\label{eq:da.1}
\end{subequations}
As we will see below, almost all primal variables in $x(t)$ are updated 
according to \eqref{eq:pma.1a} except $\omega(t)$ which is updated
according to \eqref{eq:da.1a}.

\vspace{0.1in}
\noindent \textbf{Controller \eqref{eq:control.1} design.}
Let $\lambda$ and $\mu$ be the Lagrange multipliers associated with
constraints \eqref{eq:balance.node} and \eqref{eq:opt.1a} respectively
and let $\rho := (\lambda, \mu)$.
Define the Lagrangian of \eqref{eq:opt.1} as:
\begin{align}
	L_1(x; \rho) & =   
	\frac{1}{2} \sum\nolimits_j \alpha_j \left(P^g_j\right)^2 + \frac{1}{2} \sum\nolimits_j \beta_j \left(P^l_j\right)^2
	+ \frac{1}{2} \sum\nolimits_j D_j \omega_j^2
	\nonumber 
	\\ 
	&   + \sum\nolimits_j \lambda_j \! \left( P^g_j - P^l_j - p_j \right)
	\nonumber \\
	&  
	+ \sum\nolimits_j \mu_j \bigg(
	P^g_j - P^l_j  - p_j -  D_j \omega_j  
	\nonumber \\
	&  \qquad\quad 
	+ \sum\nolimits_{i: i\rightarrow j}  B_{ij} \theta_{ij} - \! \sum\nolimits_{k: j\rightarrow k} \! B_{jk}\theta_{jk}  \!\bigg)
	\label{eq:defL.1}
\end{align}
The Lagrangian is defined to be only a function of $(x, \rho)$ and independent
of $u := (u^g, u^l)$ as we treat $u$ as a function of $(x, \rho)$
defined by the right-hand side of \eqref{eq:control.1a}\eqref{eq:control.1b}.
The set $X$ in \eqref{eq:pma.1a} is defined by the constraints 
\eqref{eq:OpConstraints.1}:
\bq
\!\!\!\!\!
X & \!\!\!\!\ := \!\!\!\! & \left\{(P^g, P^l) : 
                (\underline{P}^g, \underline{P}^l)  \ \leq \ (P^g, P^l) \ \leq \
                (\overline{P}^g, \overline{P}^l)  \right\}
\label{eq:defX}
\eq

We now explain how the closed-loop system  
\eqref{eq:model.1}\eqref{eq:control.1}  implements an approximate primal-dual
algorithm for solving \eqref{eq:opt.1} in real time.
We first show that the control \eqref{eq:control.1c} and the swing dynamic 
\eqref{eq:model.1b} implement the dual update \eqref{eq:pma.1b} on dual variables
$\rho = (\lambda(t), \mu(t))$.
We then show that \eqref{eq:model.1a}\eqref{eq:model.1c}\eqref{eq:model.1d} 
implement a mix of the primal updates \eqref{eq:pma.1a} and \eqref{eq:da.1b}
on the primal variables $x = (\theta(t), \omega(t), P^g(t), P^l(t))$.

First the variable $\lambda$ is the Lagrange multiplier vector for the per-node
power balance constraint \eqref{eq:balance.node}.  The control law
\eqref{eq:control.1c} implements part of the dual update \eqref{eq:pma.1b} in 
continuous time:
\begin{subequations}
	\bq
	\dot \lambda & = & \Gamma^{\lambda} \, \nabla_\lambda L_1 (x(t), \rho(t))
	\label{eq:dual.1a}
	\eq
	where $\Gamma^{\lambda} := \diag(\gamma^{\lambda}_j, j\in N)$.
	
	The variable $\mu$ is the Lagrange multiplier vector for the constraint
	\eqref{eq:opt.1a}.   	
	 It can be identified with the frequency deviation 
	$\omega$ as the KKT condition \cite{Bertsekas:Nonlinear}
	\bqn
	\frac{\partial L_1}{\partial \omega_j}(x^*, \rho^*) & = & 
	D_j ( \omega^*_j - \mu^*_j) \ \ = \ \ 0
	\eqn
	implies $\mu^*_j = \omega^*_j$ at optimality since $D_j>0$.   Moreover we can
	identify $\mu(t) \equiv \omega(t)$ during transient if we update the cyber quantity $\mu(t)$
	according to
	\bq
	\dot\mu & = & M^{-1} \left( P^g(t) - P^l(t)  - p_j(t) - U(\theta(t), \omega(t)) \right)
	\nonumber \\
	& = & M^{-1}\ \nabla_\mu L_1 (x(t), \rho(t))
	\label{eq:dual.1b}
	\eq
\label{eq:dual.1}
\end{subequations}
	where $M := \diag(M_j, j\in N)$.  Then $\mu$ and $\omega$ have the
	same dynamics (compare with \eqref{eq:model.1b}) and hence 
	$\mu(t) \equiv \omega(t)$ as long as $\mu(0) = \omega(0)$.
	Therefore the swing dynamic \eqref{eq:model.1b} is equivalent to \eqref{eq:dual.1b}
	and carries out the dual update \eqref{eq:pma.1b} on $\mu$ when we
	take $\mu(t)\equiv \omega(t)$.
	
	Second, to see how  \eqref{eq:model.1a}\eqref{eq:model.1c}\eqref{eq:model.1d} 
	implement the primal updates, note that the last
	term in the definition \eqref{eq:defL.1} of the Lagragian $L_1$ is:
	\bqn
	& &  \sum\nolimits_j \mu_j \left( 
	\sum\nolimits_{i: i\rightarrow j}  B_{ij} \theta_{ij} - \! \sum\nolimits_{k: j\rightarrow k} \! B_{jk}\theta_{jk}  \!\right)
	\\
	& = & 
	- \sum\nolimits_{(i,j)\in E} B_{ij} \left(\mu_i - \mu_j\right) \left(\theta_i - \theta_j\right)
	\ \ = \ \ - \mu^T C B C^T \theta
	\eqn
	We fix $\theta_0 := 0$ to be a reference angle.   Then there is a bijection
	between $\theta$ and $\tilde\theta$ that is in the column space of $C^T$, given by
	$\tilde\theta = C^T\theta$.  Hence we can work with either variable.  
	For stability proof we use $\tilde\theta$.	
	In vector form
	\bqn
	L_1 & = & \frac{1}{2} \left( (P^g)^T A^g P^g + (P^l)^T A^l P^l + \omega^T D \omega \right)
	\\
	& & + \ \lambda^T \!\! \left( P^g - P^l - p \right) + \ \mu^T \!\! \left( P^g - P^l - p -D\omega - C B\tilde \theta \right)
	\eqn
	where  $A^g := \diag(\alpha_j, j\in N)$, $A^l := \diag(\beta_j, j\in N)$, $B:=\diag(B_{ij}, (i,j)\in E)$ and
	\bqn
	\nabla_{\tilde\theta} L_1 & = & - BC^T \mu  \ \ = \ \ - BC^T \omega
	\eqn
Since
$\dot{\tilde\theta} = C^T\dot\theta = C^T\omega$, we have
\begin{subequations}
	\bq
	   \dot {\tilde \theta} & = & - B^{-1}\nabla_{\tilde\theta} L_1
	\label{eq:primal.1c}
	\eq
	i.e., \eqref{eq:model.1a} implements the primal update \eqref{eq:pma.1a} on
	$\tilde\theta$.

	Identification of $\omega(t)$ with $\mu(t)$ means that, given the dual variable
	$\rho(t)$, we update $\omega(t)$ as in the dual algorithm \eqref{eq:da.1b}:
	\bq
	\label{eq:minomega}
	\omega(t) & = & \mu(t) \ \ = \ \ \arg\min_{\omega}\, \nabla_\omega\,  L_1(x, \rho(t))
	\eq
	instead of \eqref{eq:pma.1a}.
	Moreover we have 	
	\bqn
	\nabla_{P^g} L_1 (x(t), \rho(t)) & = & A^g P^g(t) + \omega(t) + \lambda(t)
	\eqn
	Therefore the control law \eqref{eq:control.1a} is equivalent to
	\bqn
	u^g(t)  & = & \left[ P^g(t) - \, \Gamma^g\, \nabla_{P^g} L_1 (x(t), \rho(t)) 
	\right]_{\underline P^g}^{\overline P^g} \ + \ R^{-1} \omega(t)
	\eqn
	where $\Gamma^g := \diag(\gamma^g_j, j\in N)$ and 
	$R := \diag(R_j, j\in N)$.	
	Then the generation dynamic \eqref{eq:model.1c} becomes
	\begin{align}
	T^g\dot P^g  =  \left[ P^g(t) - \, \Gamma^g\, \nabla_{P^g} L_1 (x(t), \rho(t))
	\right]_{\underline P^g}^{\overline P^g} \ - \ P^g(t)
	\label{eq:primal.1a}
	\end{align}
	where $T^g:=\diag(T^g_j, j\in N )$.
	Similarly the control law \eqref{eq:control.1b} is equivalent to
	\bqn
	u^l(t)  & = & \left[ P^l(t) - \, \Gamma^l\, \nabla_{P^l} L_1 (x(t), \rho(t)) 
	\right]_{\underline P^l}^{\overline P^l}
	\eqn
	where $\Gamma^l := \diag(\gamma^l_j, j\in N)$.  	
	The controllable load dynamic \eqref{eq:model.1d} is equivalent to
	\begin{align}
	T^l\dot P^l & =  \left[ P^l(t) - \, \Gamma^l\, \nabla_{P^l} L_1 (x(t), \rho(t))
	\right]_{\underline P^l}^{\overline P^l} \ - \ P^l(t)
	\label{eq:primal.1b}
	\end{align}
	where $T^l:=\diag(T^l_j, j\in N )$.
	
	Writing $P:=(P^g, P^l)$, $T^{gl} = \text{diag}(T^g, T^l)$ and $\Gamma^{gl} = \text{diag}(\Gamma^g, \Gamma^l)$, 
	the dynamics \eqref{eq:primal.1a}--\eqref{eq:primal.1b} becomes
	\bq
	\!\!\!\!\!\!
	T^{gl}\dot P & \!\!\!\!\! = \!\!\!\!\!  & \text{Proj}_X\!\! \left( P(t) - \, \Gamma^{gl}\, \nabla_{P} L_1 (x(t), \rho(t)) \right)
			\, - \, P(t)
	\label{eq:primal.1ab}
	\eq
	where $X$ is defined in \eqref{eq:defX}.
	Informally \eqref{eq:primal.1ab} can be interpreted as a continuous-time 
	version of the primal update \eqref{eq:pma.1a} since the right-hand side can
	be interpreted as $P(t+1) - P(t)$ in the discrete-time version \eqref{eq:pma.1a}.
	While it is clear from \eqref{eq:pma.1a} that $P(t)$ in the discrete-time formulation
	stays in $X$ for all $t$, it may not be obvious that $P(t)$ in the continuous-time 
	formulation \eqref{eq:primal.1ab} stays in $X$ for all $t$.  This is
	proved formally in Lemma \ref{lemma:P(t)} below.
\label{eq:primal.1}
\end{subequations}

	In summary the closed-loop system \eqref{eq:model.1}\eqref{eq:control.1} 
	carries out an approximate primal-dual algorithm \eqref{eq:pma.1} in
	continuous time.  The dual updates \eqref{eq:dual.1a} and \eqref{eq:dual.1b}
	on $(\lambda(t), \mu(t))$
	are implemented by \eqref{eq:control.1c} and \eqref{eq:model.1b} respectively.
	The primal updates \eqref{eq:primal.1c} and \eqref{eq:primal.1ab} on
	$(\theta(t), P^g(t), P^l(t))$ are
	implemented by \eqref{eq:model.1a} and \eqref{eq:model.1c} \eqref{eq:model.1d}
	respectively.
	We refer to this as an \emph{approximate} primal-dual algorithm because
	the identification of $\omega(t)\equiv \mu(t)$ implements the update
	\eqref{eq:da.1b} on $\omega(t)$ instead of \eqref{eq:pma.1a}.

\subsection{Optimality and uniqueness of equilibrium point}
\label{subsec:optimality.1}

In this subsection, we address the optimality of the equilibrium point of the closed-loop system \eqref{eq:model.1}\eqref{eq:control.1}. Given an $(x, \rho) := ( \theta, \omega, P^{g}, P^{l},$ $\lambda, \mu)$, recall that the control input $u(x, \rho)$ is given by \eqref{eq:control.1'}.
\begin{definition}
	\label{def:ep.1}
	A point $(x^*, \rho^*) := ( \theta^*, \omega^*, P^{g*}, P^{l*},$ $\lambda^*, \mu^*)$
	is an \emph{equilibrium point} or an \emph{equilibrium} of the closed-loop system 
	\eqref{eq:model.1}\eqref{eq:control.1} if 
	\bee
	\item The right-hand side of \eqref{eq:model.1} vanishes at $x^*$ and $u(x^*, \rho^*)$.  
	\item The right-hand side of \eqref{eq:control.1c} vanishes at $(x^*, \rho^*)$.
	\eee
\end{definition}

\begin{definition}
	A point $(x^*, \rho^*)$ is \emph{primal-dual optimal} if $(x^*, u(x^*, \rho^*))$ is optimal 
	for \eqref{eq:opt.1} and $\rho^*$ is optimal for its dual problem.
\end{definition}

Section \ref{subsec:design} says that the closed-loop system \eqref{eq:model.1}\eqref{eq:control.1}
carries out an (approximate)
 primal-dual algorithm in real time to solve \eqref{eq:opt.1}.
In this subsection we prove that a point $(x^*, \rho^*)$ is an equilibrium of the closed-loop system if 
and only if it is primal-dual optimal.   Moreover the equilibrium is unique.
In the next subsection we prove that the closed-loop system converges to the equilibrium point
starting from any initial point that satisfies constraint \eqref{eq:OpConstraints.1}. 

\begin{theorem}
	\label{thm:2}
	Suppose assumption A2 hold.   A point  $(x^*, \rho^*)$ is primal-dual optimal if
	and only if  $(x^*, \rho^*)$ is an equilibrium of the closed-loop 
	system \eqref{eq:model.1}\eqref{eq:control.1} that satisfies 
	\eqref{eq:OpConstraints.1} and $\mu^*=0$.
\end{theorem}

Theorem \ref{thm:2} shows the equivalence between
the equilibrium of closed-loop system and the primal-dual optimal solution. 
It also implies that, in equilibrium, per-node power 
balance \eqref{eq:balance.node} is achieved and 
constraints \eqref{eq:OpConstraints.1} are satisfied.   
The next theorem shows that the equilibrium point is almost unique
and has a simple and intuitive structure. 
\begin{theorem}
	\label{thm:1}
	Suppose assumption A1 and A2 hold.  Let $(x^*, \rho^*)$ be primal-dual optimal. 
	Then
	\begin{enumerate}
		\item $x^*$ and $\mu^*$ are unique, with $\theta^*$ being unique
		up to an (equilibrium) reference angle $\theta_0^*$.   
		\item $\lambda^*$ is also unique if strict inequalities hold in A2.  In that case,
		$\lambda^*_j$ equals the (negative of the) marginal generation/load regulation cost at
		node $j$, i.e., $\alpha_j P^{g*}_j = -\beta_j P^{l*}_j = -\lambda_j^*$.
			
		\item nominal frequencies are restored, i.e., $\omega^*_j=0$ for all $j\in N$;
		moreover $\theta^*_j=\theta^*_0$ for all $j\in N$.
		\item the power flow $P^*_{ij} := B_{ij}(\theta^*_i - \theta^*_j) = 0$ on every line $(i,j)\in E$.
	\end{enumerate}
\end{theorem}

The proofs of Theorem \ref{thm:2} and \ref{thm:1} are given in Appendix 
\ref{sec:ProofThms1and2}.

\subsection{Asymptotic stability}
Before proving the stability, we assume
\bi
\item[\textbf{A3:}] The initial state of the closed-loop system  
	\eqref{eq:model.1}\eqref{eq:control.1} is finite, and $(P^g_j(0), P^l_j(0))$ 
	satisfy constraint \eqref{eq:OpConstraints.1}. 
\ei

Motivated by \eqref{eq:primal.1ab} we will write the closed-loop system 
\eqref{eq:model.1}\eqref{eq:control.1} in a similar form that will turn out to be
critical for our stability analysis.   To do this we first prove the following boundedness property
of $(P^g(t), P^l(t))$ in Appendix \ref{proof:thmstability2}.
\begin{lemma}
\label{lemma:P(t)}
Suppose assumptions A1 and A3 hold.  Then constraint \eqref{eq:OpConstraints.1} is 
satisfied for all $t\geq 0$, i.e. $(P^g(t), P^l(t))\in X$ for all $t\geq 0$ where
$X$ is defined in \eqref{eq:defX}.
\end{lemma}

We set the control gains for $(\hat u^g, \hat u^l)$ in \eqref{eq:control.1} as 
	$ \gamma^g_j \ = \ ( T^g_j )^{-1}, 
	\gamma^l_j \ = \ ( T^l_j )^{-1}
	$
Identifying $\mu(t)\equiv \omega(t)$, the closed-loop system
\eqref{eq:model.1}\eqref{eq:control.1} is (in vector form): 
\begin{subequations}
	\bq
	\dot {\tilde\theta}(t)&=&    C^T\omega(t) \label{eq:model.2a'}\\
	\dot \omega (t)&=&M^{-1}\left(P^g(t)-P^l(t)-p(t)-D\omega(t)-CB\tilde\theta(t) \right ) 
	\nonumber \label{eq:model.2b2'}\\
	\\
	\dot{P}^g(t)&=&(T^g)^{-1}\left ( -P^g(t)+\hat u^g(t) \right)\\ 
	\dot{P}^l(t)&=&(T^l)^{-1}\left ( -P^l(t)+\hat u^l(t) \right)\\
	\dot{\lambda}(t)&=&\Gamma^{\lambda} \left (P^g(t)-P^l(t)-p \right)
	\eq
Here
\label{eq:model.2}
\end{subequations}
\bqn
\hat u^g(t) & = & \left[ P^g(t) - (T^g)^{-1} \left( 
		A^g P^g(t) + \omega(t) + \lambda(t) \right)\right]_{\underline P^g}^{\overline P^g}
\\
\hat u^l(t) & = & \left[ P^l(t) - (T^l)^{-1} \left( 
		A^l P^l(t) - \omega(t) - \lambda(t) \right)\right]_{\underline P^l}^{\overline P^l}
\eqn
Denote $w:=(\tilde\theta, \omega, P^g, P^l, \lambda)$ and define
\bq
F(w)& \!\!\!\!\! := \!\!\!\!\! &\left [ 
\begin{array}{l}
	- B^{{1}/{2}} C^T \omega \\
	- M^{-1/2} \left(P^g-P^l-p-D\omega-CB\tilde\theta \right )\\ 
	(T^g)^{-1} \left(A^g P^g+\omega+\lambda \right)\\
	(T^l)^{-1} \left( A^l P^l-\omega-\lambda \right)\\ 
	- (\Gamma^\lambda)^{{1}/{2}} \left (P^g-P^l -p \right)
\end{array}  \right ] 
\label{eq:Fz.1}
\eq
We further define
\bqn
S  & := & \mathbb R^{m+n+1}\times X \times \mathbb R^n
\eqn
where the closed convex set $X$ is defined in \eqref{eq:defX}.  
For any $w$ denote the projection of $w-F(w)$ onto $S$ to be
\begin{align}
H(w) :=  \text{Proj}_S (w-F(w)) := \ \arg\min_{y\in S} \| y - (w-F(w)) \|_2 \nonumber
\end{align}
where $\|\cdot\|_2$ is the Euclidean norm.
Then the closed-loop system \eqref{eq:model.2} can be written as
\bq
\dot w(t) & = & \Gamma_1 \left( H(w(t)) \ - \ w(t) \right)
\label{eq:model.3}
\eq
where the positive definite gain matrix is:
\bqn
\Gamma _1& := & \text{diag}\left( 
		B^{-1/2}, M^{-1/2}, (T^g)^{-1}, (T^l)^{-1}, (\Gamma^\lambda)^{1/2} \right)
\eqn
Note that the projection operation $H$ has an effect only on 
$(\dot P^g, \dot P^l)$.   Lemma \ref{lemma:P(t)} implies that $w(t) \in S$ 
for all $t$, justifying the equivalence of \eqref{eq:model.2} and \eqref{eq:model.3}.

A point $w^*\in S$ is an \emph{equilibrium} of the closed-loop system \eqref{eq:model.3}
if and only if it is a fixed point of the projection:
\bqn
H(w^*) & = & w^*
\eqn
Let $E_1:=\{w \ |\ H(w(t))-w(t)=0 \}$ be the set of equilibrium points. Then we have the following theorem.
\begin{theorem}
	\label{thm:stability.2}
	Suppose A1, A2 and A3 hold.   Starting from any initial point $w(0)$, 
	$w(t)$ remains in a bounded set for all $t$ 
	and $w(t)\rightarrow w^*$ as $t\rightarrow\infty$ for some equilibrium $w^*\in E_1$ that
	is optimal for problem (\ref{eq:opt.1}).
	If strictly inequalities hold in A2, then the equilibrium point $w^*$ of the closed-loop 
	system \eqref{eq:model.3} is unique.
\end{theorem}
The theorem implies that if strict inequalities hold in A2, then,
starting from any initial point $w(0)$, the trajectory 
$w(t)$ of the closed-loop system \eqref{eq:model.3} converges to the unique equilibrium 
$w^*$ as $t\to \infty$.


We comment on the proof of the theorem given in Appendix \ref{proof:thmstability2}.
Unlike the quadratic Lyapunov function used 
in \cite{arrow1958studies, feijer2010stability, rantzer2009dynamic, Changhong:Design, 
LiZhaoChen2016, Mallada-2017-OLC-TAC} for the analysis of primal-dual algorithms,
we use the following Lyapunov function: 
\begin{align}
	\label{eq:lyapunov.1}
	V_1(w) & =  - \left( H(w)-w \right)^T F(w)  \, - \, \frac{1}{2} ||H(w)-w||^2_2 \nonumber\\
	&\quad +\frac{1}{2}k(w-w^*)^T\Gamma_1^{-2}(w-w^*) 
\end{align}
where $w^*$ is an equilibrium point (to be determine later) and $k>0$ is 
small enough $k>0$ such that the diagonal matrix 
$\Gamma_1-k\Gamma_1^{-1} > 0$, i.e., is strictly positive definite.
The first part of $V_1$ is motivated by the observation in 
\cite{Fukushima:Equivalent} that $H(w)-w$ with a stepsize computed from
an exact line search defines an iterative descent algorithm for minimizing 
the following function over $S$:
\bqn
	\hat V_1(w) & = & - \left( H(w)-w \right)^T F(w)  \, - \, \frac{1}{2} ||H(w)-w||^2_2 
\eqn
It is proved in \cite[Theorem 3.1]{Fukushima:Equivalent}  that $\hat V_1(w) \geq 0$ 
on $S$ and $\hat V_1(w)=0$ holds only at any equilibrium $w^* = H(w^*)$.
The use of $\tilde\theta$ instead of $\theta$ in \eqref{eq:model.2} and the 
definitions of $F$ and $\Gamma_1$ in \eqref{eq:Fz.1}\eqref{eq:model.3}
are carefully chosen in order to prove that
$\dot V_1(w(t))\leq 0$ along any solution trajectory.   
The second part 
\bqn
\frac{1}{2}k(w-w^*)^T\Gamma_1^{-2}(w-w^*) 
\eqn
of $V_1$ is motivated by the quadratic Lyapunov function used
in \cite{arrow1958studies, feijer2010stability, rantzer2009dynamic, Changhong:Design, 
LiZhaoChen2016, Mallada-2017-OLC-TAC} for the analysis of primal-dual algorithms.
While the first part $\hat V_1$ is critical for proving $\dot V_1\leq 0$, implying that
any trajectory $w(t)$ of the closed-loop system converges to a set of equilibrium
points by LaSalle's invariance principle, the quadratic term
$(w- w^*)^T\Gamma_1^{-2}(w- w^*)$ in $V_1$ is used to prove 
that $w(t)$ actually converges to a limit point, using the technique due
to \cite{Changhong:Design, LiZhaoChen2016}.

%% file: sec-sim_part_I.tex
\section{Case studies}
\subsection{System configuration}
To test the optimal frequency controller, we modify Kundur's four-machine, two-area system \cite{Fang:Design} \cite{Kundur:Power} by expanding it to a four-area system. Each area has one (aggregate) generator (Gen1$\sim$Gen4), one controllable (aggregate) load (L1c$\sim$L4c) and one uncontrollable (aggregate) load (L1$\sim$L4), as shown in Fig.\ref{fig:system}. The parameters of generators and controllable loads are given in Table \ref{tab:SysPara}. For others one can refer to \cite{Kundur:Power}. The total uncontrollable load in each area are identically 480MW. At time $t=10s$, we add step changes on the uncontrollable loads in four areas to test the performance of our controllers. 

All the simulations are implemented in PSCAD \cite{website:PSCAD}
with 8GB memory and 2.39 GHz CPU.  The detailed electromagnetic transient model of three-phase synchronous machines is adopted to simulate generators with both governors and exciters.  The uncontrollable loads L1-L4 are modelled by the fixed load in PSCAD, while controllable load L1c-L4c are formulated by the self-defined controlled current source. The closed-loop system diagram is shown in Fig.\ref{fig:control}. We only need measure loacal frequency, generation, controllable load and tie-line power flows to compute control demands. There are no need of uncontrollable load and communication from other areas.
Note that in the simulation, all 
variables are added by their initial steady state values to explicitly show the actual values.
\begin{figure}[htp]
        \centering
        \includegraphics[width=0.4\textwidth]{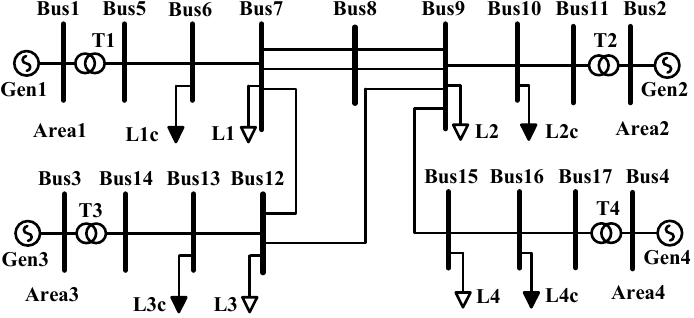}
        \caption{Four-area power system}
        \label{fig:system}
\end{figure}

\begin{figure}[htp]
	\centering
	\includegraphics[width=0.35\textwidth]{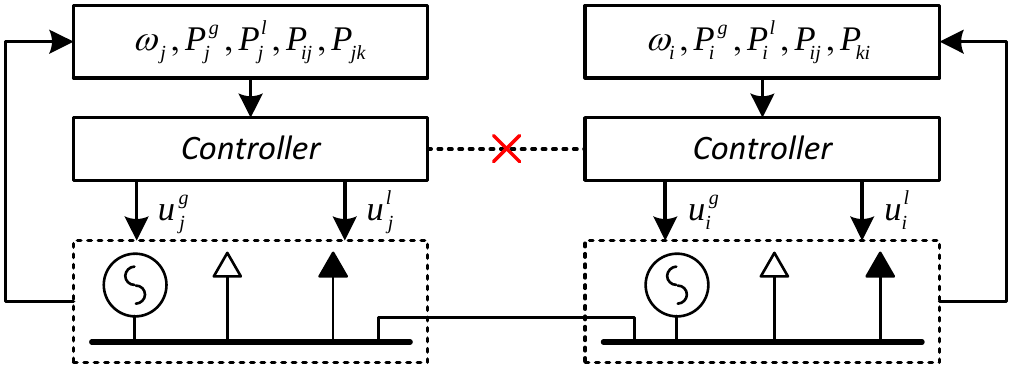}
	\caption{Closed-loop system diagram}
	\label{fig:control}
\end{figure}

\begin{table}[http]
        \centering
        \caption{\\ \textsc{System parameters}}
        \label{tab:SysPara}
        \begin{tabular}{c c c c c c c}
                \hline 
                Area $j$ & $D_j$ & $R_j$ & $\alpha_j$ & $\beta_j$ &$T^g_j$ &$T^l_j$\\
                \hline
                1 & 0.04  & 0.04  & 2    & 2.5  & 4    & 4 \\
                2 & 0.045  & 0.06  & 2.5  & 4    & 6    & 5 \\
                3 & 0.05  & 0.05  & 1.5  & 2.5  & 5    & 4 \\
                4 & 0.055  & 0.045 & 3    & 3    & 5.5  & 5 \\
                \hline
        \end{tabular}
\end{table}

\subsection {Simulation Results}

In the simulation, the generations in each area are (625.9, 562.7, 701.7, 509.6) MW and the controllable loads are all 120 MW. The load changes are given in Table \ref{tab:DistConstraints}, which are unknown to the controllers. Here we use the method mentioned in Remark 2 to estimate the load change. We also show the operational constraints on generations and controllable loads in individual control areas in Table \ref{tab:DistConstraints}. 

\begin{table}[htb]
	\centering
	\caption{ \\ \textsc {Capacity limits and load disturbance }}
	\label{tab:DistConstraints}
	\begin{tabular}{c c c c}
		\hline 
		Area $j$ & Load changes & [$\underline{P}^g_j$, $\overline{P}^g_j$] (MW) & [$\underline{P}^l_j$, $\overline{P}^l_j$] (MW)\\
		\hline
		1 & 90 MW  & [600, 730] & [75, 120]\\
		2 & 90 MW  & [550, 680] & [80, 120]\\
		3 & 90 MW  & [650, 810] & [80, 120]\\
		4 & 120 MW & [500, 640] & [55, 120]\\
		\hline
	\end{tabular}
\end{table}

\subsubsection {Stability and optimality}
The dynamics of local frequencies  and tie-line power flows are illustrated in  Fig.\ref{fig:stab}. Both the frequency and tie-line power deviations are restored in all the four control areas. The generations and controllable loads are different from those before disturbance, indicating that the system is stabilized at a new equilibrium point. The resulting equilibrium point is  given in Table \ref{tab:eper}, which is identical to the optimal solution of  \eqref{eq:opt.1} computed by centralized optimization using CVX. The simulation results confirm our theoretic analyses, verifying that our controller can autonomously guarantee the frequency stability while achieving optimal operating point in a completely decentralized manner.

\begin{figure}[htp]
        \centering
        \includegraphics[width=0.5\textwidth]{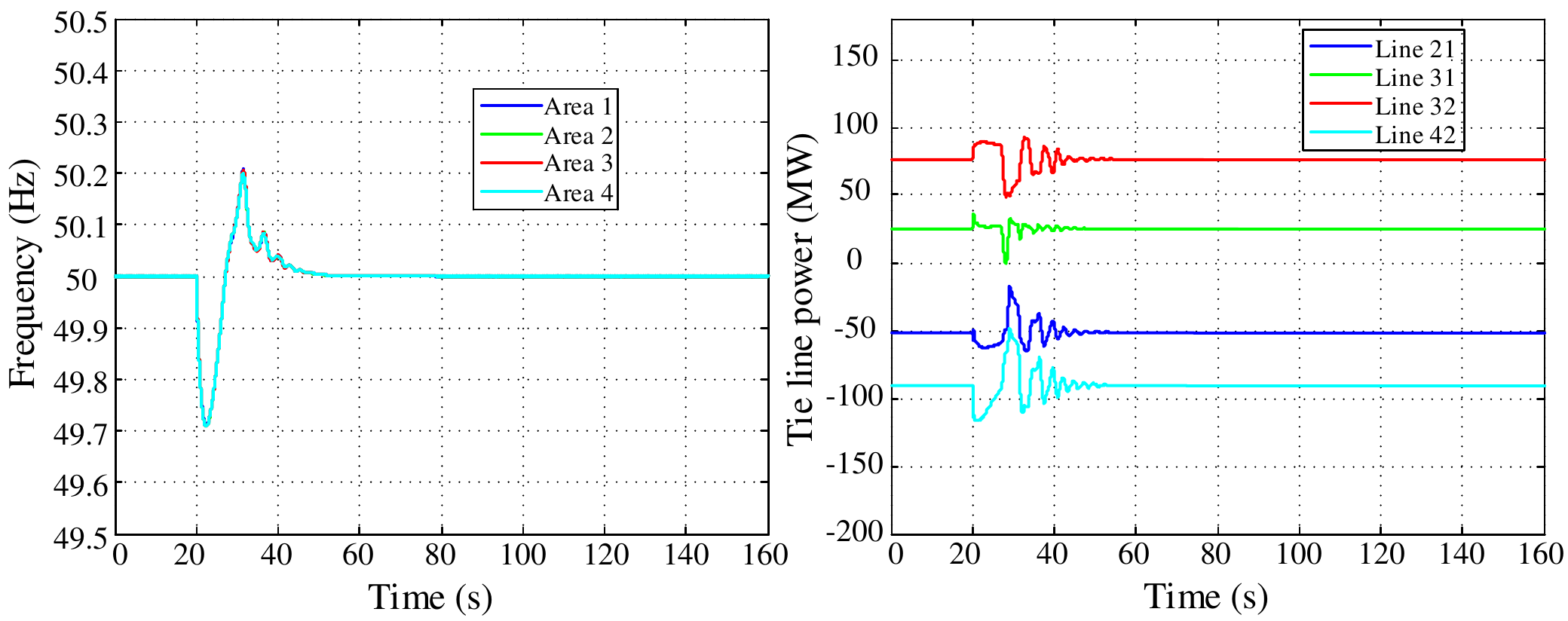}
        \caption{Dynamics of frequencies (left) and tie-line flows (right) in per-node balance case}
        \label{fig:stab}
\end{figure}

\begin{table}[htb]
        \centering
        \caption{\\ \textsc {Equilibrium points}}
        \label{tab:eper}
        \begin{tabular}{c c c c c}
                \hline 
                & Area 1 & Area 2 & Area 3& Area 4\\
                \hline
                $P^{g*}_j $ (MW) & 676 & 618  & 758  & 570\\
                $P^{l*}_j $ (MW) & 80  & 85.3 & 86.2 & 60 \\
                \hline
        \end{tabular}
\end{table}

\subsubsection{Dynamic performance}
In this subsection, we analyze the impacts of regulation capacity constraints on the dynamic performance. To  this end, we compare the dynamic responses of the frequency controllers with and without input saturations. The trajectories of mechanical powers of turbines and controllable loads are shown in Fig.\ref{fig:dynamic.mec} and Fig.\ref{fig:dynamic.2}, respectively. In this case, the system frequency and tie-line flows are restored, and the same optimal equilibrium point is achieved. With the saturated   controller, the mechanical power of turbines and controllable loads are  strictly within the limits in  transient. On the contrary, the controller without saturation results in considerable violation of the capacity constraints during transient, which is practically infeasible. 


\begin{figure}[!t]
	\centering
	\includegraphics[width=0.5\textwidth]{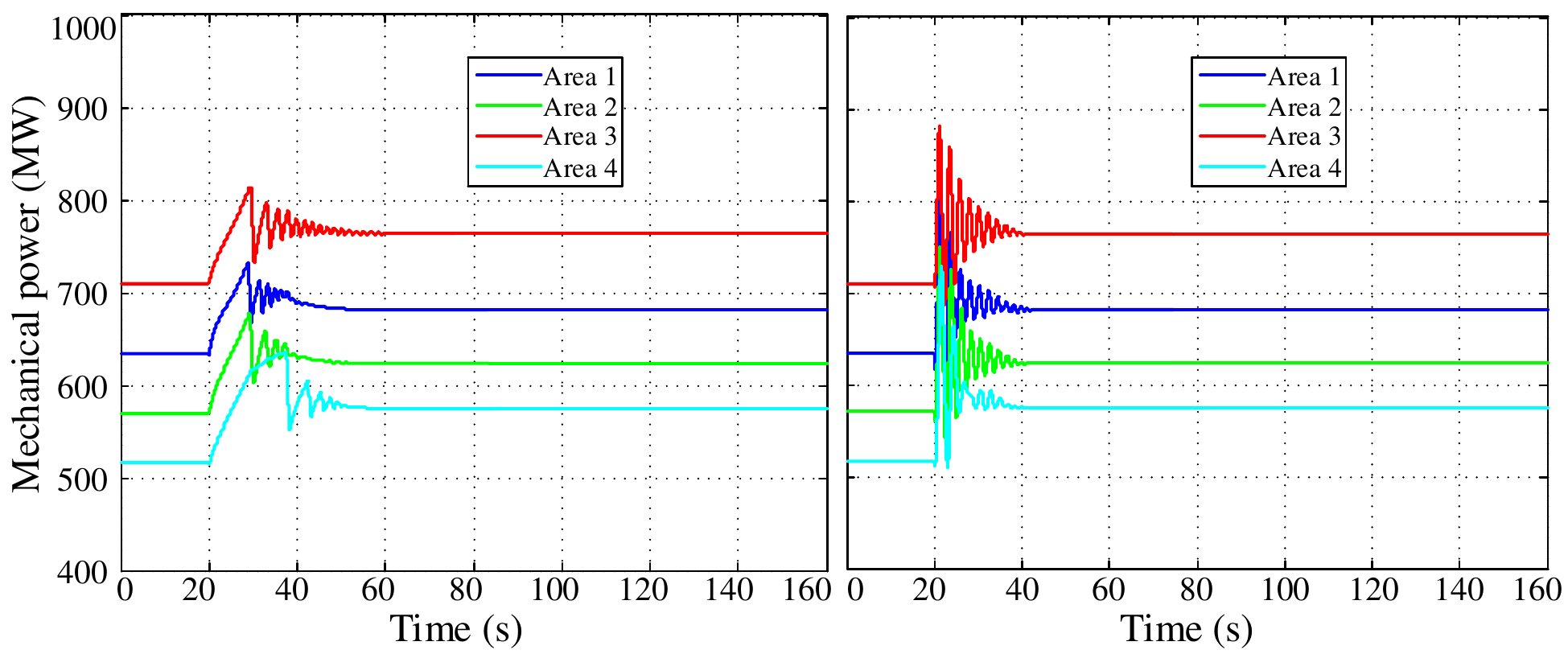}
	\caption{Mechanical outputs with(left)/without(right) capacity constraints}
	\label{fig:dynamic.mec}
\end{figure}


\begin{figure}[!t]
        \centering
        \includegraphics[width=0.5\textwidth]{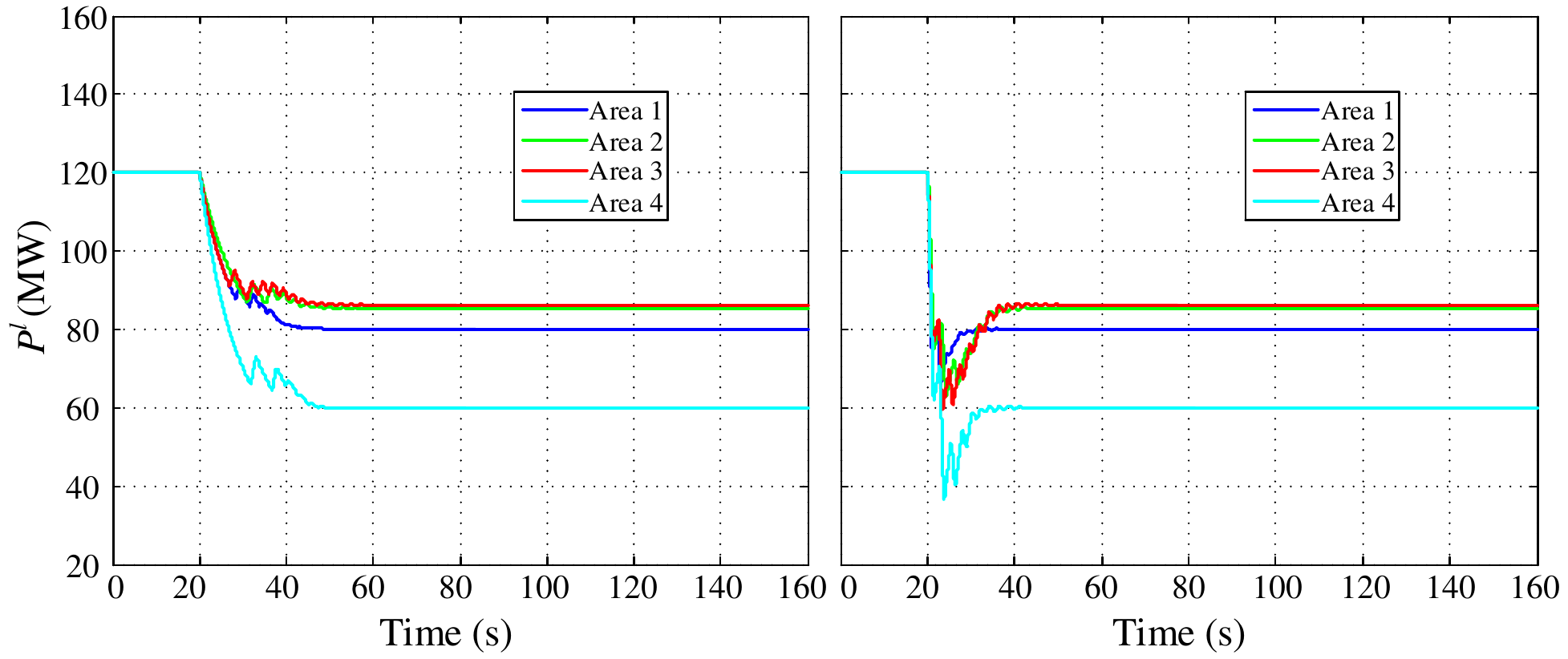}
        \caption{Controllable loads with(left)/without(right) capacity constraints}
        \label{fig:dynamic.2}
\end{figure}

%% file: sec-conclusion_part_I.tex
\section{Conclusion}
We have proposed a decentralized optimal frequency control with aggregate generators and controllable loads. The controller can autonomously restore the nominal frequencies and tie-line powers after unknown load disturbances while minimizing the regulation cost. The capacity constraints on the generations and the controllable loads can always be satisfied even during transient. We have revealed that the closed-loop system carry out a primal-dual algorithm to solve the associated optimal dispatch problem, guaranteeing the optimality of closed-loop equilibria. We have used the projection technique combined with LaSalle's invariance principle to prove the asymptotically stability of the closed-loop system. Simulations on the modified Kundur's power system verify the efficacy of our designs. 

%

%% file: sec-appendix_part_I.tex
\newpage

\appendices

\makeatletter
\@addtoreset{equation}{section}
\@addtoreset{theorem}{section}
\makeatother
\renewcommand{\theequation}{A.\arabic{equation}}
\renewcommand{\thetheorem}{A.\arabic{theorem}}

\section{Proofs of Theorem \ref{thm:2} and Theorem \ref{thm:1}}
\label{sec:ProofThms1and2}

We start with a lemma.
\begin{lemma}
	\label{lemma:1}
	Suppose $(x^*, u^*)$ is optimal for \eqref{eq:opt.1}.  Then $\omega^*=0$
	and $\theta^*  = \theta^*_0\, \textbf{1}$ where $\textbf{1}$ is the vector
	with all entries being 1.
\end{lemma}
\begin{proof}
	Suppose for the sake of contradiction that $\omega^*\neq 0$.  
	Construct from $x^*$ another
	point $\hat x$ by setting $\hat\theta=0$, $\hat \omega=0$ and keeping
	the other components of $x^*$ unchanged.  
	Since $x^*$ satisfies both \eqref{eq:balance.node} and \eqref{eq:opt.1a}
	we must have $U(\theta^*, \omega^*)=0$.  This also holds for $\hat x$,
	i.e., $U(\hat\theta, \hat\omega)=0$, and hence $(\hat x, u^*)$ remains feasible
	since other components of $\hat x$ are the same as those of $x^*$.
	Moreover $(\hat x, u^*)$ has a strictly lower cost than $(x^*, u^*)$, 
	contradicting the optimality of $(x^*, u^*)$.
	Hence any optimal $(x^*, u^*)$ must have $\omega^*=0$.
	
	We claim that $\omega^* = 0$ implies that $\theta^*=\theta^*_0\, \textbf{1}$.  
	For any feasible point $x$,     \eqref{eq:balance.node} and \eqref{eq:opt.1a} 
	imply that 
	\bqn
	U(\theta, \omega) & = & D\omega + CBC^T \theta \ \ = \ \ 0
	\eqn
	Hence we have $CBC^T \theta^* = 0$ at an optimal $x^*$.  Since $CBC^T$ is 
	an $(n+1)\times (n+1)$ matrix with rank $n$, its null space has dimension 1.
	The vector $\textbf{1}$ is in its null space because $C^T\textbf{1} = 0$.
	Hence $\theta^*=\theta^*_0\, \textbf{1}$.
\end{proof}

Lemma \ref{lemma:1} at optimality implies that the frequencies are restored. Noticing that  $\theta^*  = \theta^*_0\, \textbf{1}$ means all angles are equal, such an optimal solution also implies 
that all the tie-line power flows are restored to their nominal values. 

\begin{lemma}
	\label{lemma:2}
	Suppose $(x^*, \rho^*)$ is primal-dual optimal.  Then
	\bqn
	u^{g*}_j \ = \
	P^{g*}_j & = & \left[ 
	P^{g*}_j - \gamma^g_j \left( \alpha_j P^{g*}_j + \omega^*_j + \lambda_j^* \right)
	\right]_{\underline P^g_j}^{\overline P^g_j}
	\nonumber 	\\
	u^{l*}_j \ = \
	P^{l*}_j & = & \left[ 
	P^{l*}_j - \gamma^l_j \left( \beta_j P^{l*}_j - \omega^*_j - \lambda_j^* \right)
	\right]_{\underline P^l_j}^{\overline P^l_j}
	\eqn
	for any $\gamma^g_j>0$ and $\gamma^l_j>0$.
\end{lemma}
\begin{proof}
	Since \eqref{eq:opt.1} is convex with linear constraints, strong duality holds.  Hence 
	$(x^*, \rho^*)$ is a primal-dual optimal if and only if it satisfies the KKT condition:
	$(x^*, u(x^*, \rho^*))$ is primal feasible and
	\begin{align}
	x^*& = \, \arg\min_x \left\{ L_1(x; \rho^*) \vert \, (x, u(x, \rho^*)) \text{ satisfies } \eqref{eq:OpConstraints.1}
	\eqref{eq:opt.1b} \eqref{eq:opt.1c} \right\}
	\label{eq:lemma.0}
	\end{align}
	From the definition \eqref{eq:defL.1} of $L_1$, $x^*$ satisfies \eqref{eq:lemma.0}
	if and only if $(x^*, u(x^*, \rho^*))$ satisfies 
	\eqref{eq:OpConstraints.1}      \eqref{eq:opt.1b} \eqref{eq:opt.1c} and the first-order
	stationarity condition, i.e., for all $j\in N$,
	\begin{subequations}
		\begin{align}
		\!\!\!\! \alpha_j P^{g*}_j + \mu^*_j+ \lambda^*_j  & \left\{ \begin{array}{lll}
		\geq & 0 & \text{ if } P^{g*}_j = \underline{P}^g_j \\
		= & 0 & \text{ if }  \underline{P}^g_j < P^{g*}_j < \overline{P}^g_j \\
		\leq & 0 & \text{ if } P^{g*}_j = \overline{P}^g_j \\
		\end{array}  \right.
		\label{eq:lemma.1a}
		\\
		\!\!\!\! \beta_j P^{l*}_j - \mu^*_j- \lambda^*_j  & \left\{ \begin{array}{lll}
		\geq & 0 & \text{ if } P^{l*}_j = \underline{P}^l_j \\
		= & 0 & \text{ if }  \underline{P}^l_j < P^{l*}_j < \overline{P}^l_j \\
		\leq & 0 & \text{ if } P^{l*}_j = \overline{P}^l_j \\
		\end{array}  \right.
		\label{eq:lemma.1b}
		\\
		D_j (\omega_j^* - \mu_j^*) & \ \ \ \, = \ \ \, 0
		\label{eq:lemma.1c}
		\\
		\sum_{i:i\rightarrow j} B_{ij} (\mu^*_i - \mu^*_j) & \ \ \ \,= \ 
		\sum_{k:j\rightarrow k} B_{jk} (\mu^*_j - \mu^*_k) 
		\end{align}
		From Lemma \ref{lemma:1} we have $\omega^*=0$ and hence \eqref{eq:lemma.1}
		reduces to $\mu^* = \omega^* = 0$ since $D_j>0$ and
		\label{eq:lemma.1}
	\end{subequations}
	\begin{subequations}
		\begin{align}
		\!\!\!\! \alpha_j P^{g*}_j  + \omega^*_j + \lambda^*_j & \left\{ \begin{array}{lll}
		\geq & 0 & \text{ if } P^{g*}_j = \underline{P}^g_j \\
		= & 0 & \text{ if }  \underline{P}^g_j < P^{g*}_j < \overline{P}^g_j \\
		\leq & 0 & \text{ if } P^{g*}_j = \overline{P}^g_j \\
		\end{array}  \right.
		\\
		\!\!\!\! \beta_j P^{l*}_j  - \omega^*_j - \lambda^*_j & \left\{ \begin{array}{lll}
		\geq & 0 & \text{ if } P^{l*}_j = \underline{P}^l_j \\
		= & 0 & \text{ if }  \underline{P}^l_j < P^{l*}_j < \overline{P}^l_j \\
		\leq & 0 & \text{ if } P^{l*}_j = \overline{P}^l_j \\
		\end{array}  \right.
		\end{align}
		\label{eq:lemma.2}
	\end{subequations}
	It can be checked that \eqref{eq:lemma.2} is equivalent to
	\bqn
	P^{g*}_j & = & \left[ 
	P^{g*}_j - \gamma^g_j \left( \alpha_j P^{g*}_j + \omega^*_j + \lambda_j^* \right)
	\right]_{\underline P^g_j}^{\overline P^g_j}
	\nonumber \\
	\\
	P^{l*}_j & = & \left[ 
	P^{l*}_j - \gamma^l_j \left( \beta_j P^{l*}_j - \omega^*_j - \lambda_j^* \right)
	\right]_{\underline P^l_j}^{\overline P^l_j}
	\eqn
	for any $\gamma^g_j>0$ and $\gamma^l_j>0$.
	The lemma then follows from \eqref{eq:opt.1b}\eqref{eq:opt.1c}.
\end{proof}


We now prove Theorems \ref{thm:2} and \ref{thm:1}.
\begin{proof}[Proof of Theorem \ref{thm:2}]
	\noindent
	$\Rightarrow$: 
	Suppose $(x^*, \rho^*)$ is primal-dual optimal. 
	Then $x^*$ satisfies the operational constraints \eqref{eq:OpConstraints.1}.
	Moreover the right-hand side of \eqref{eq:model.1} vanishes because:
	\bi
	\item $\dot\theta = 0$ since $\omega^*=0$ from Lemma \ref{lemma:1}.
	\item $\dot\omega = 0$ because of \eqref{eq:opt.1a}.
	\item $\dot P^g = \dot P^l=0$ since $\omega^*= 0$ and $x^*$ 
	satisfies \eqref{eq:opt.1b} and \eqref{eq:opt.1c}.
	\ei
	The right-hand side of \eqref{eq:control.1c} vanishes because $x^*$ satisfies 
	per-node power balance \eqref{eq:balance.node}.   By Lemma \ref{lemma:2}
	$(x^*, \rho^*)$ satisfies \eqref{eq:control.1a}\eqref{eq:control.1b}.
	Hence $(x^*, \rho^*)$ is an equilibrium of the closed-loop system 
	\eqref{eq:model.1}\eqref{eq:control.1} that satisfies 
	the operational constraints \eqref{eq:OpConstraints.1}.   Moreover 
	$\mu^* = \omega^* = 0$ by \eqref{eq:lemma.1c} since $D_j>0$ for all $j\in N$.

	\vspace{0.07in}
	\noindent
	$\Leftarrow$: Suppose now $(x^*, \rho^*)$ is an equilibrium of the closed-loop 
	system \eqref{eq:model.1}\eqref{eq:control.1} and satisfies \eqref{eq:OpConstraints.1}
	with $\mu^*=0$.
	Since \eqref{eq:opt.1} is convex with linear constraints, 
	$(x^*, \rho^*)$ is a primal-dual optimal if and only if 
	$(x^*, u(x^*, \rho^*))$ is primal feasible and satisfies \eqref{eq:lemma.0}
	(note that $\nabla_\rho L_1(x^*, \rho^*)=0$ since
	$\dot\mu = \dot\lambda = 0$).
	
	To show that $(x^*, u(x^*, \rho^*))$ is primal feasible, note that since 
	$(x^*, u(x^*, \rho^*))$ is an equilibrium of 
	\eqref{eq:model.1}, it satisfies $\omega^*=0$, 
	\eqref{eq:opt.1a}\eqref{eq:opt.1b}\eqref{eq:opt.1c}, in
	addition to \eqref{eq:OpConstraints.1}.
	Since $(x^*, \rho^*)$ is a closed-loop equilibrium, we have $\dot\lambda^* \equiv 0$ in
	\eqref{eq:control.1c}, implying \eqref{eq:balance.node}.
	Hence $x^*$ is primal feasible.
	
	To show that $(x^*, \rho^*)$ satisfies \eqref{eq:lemma.0}, note that 
	\eqref{eq:opt.1b}\eqref{eq:opt.1c} and \eqref{eq:control.1a}\eqref{eq:control.1b}
	imply that
	\bqn
	P^{g*}_j & = & \left[ 
	P^{g*}_j - \gamma^g_j \left( \alpha_j P^{g*}_j + \omega^*_j + \lambda_j^* \right)
	\right]_{\underline P^g_j}^{\overline P^g_j}
	\nonumber \\
	P^{l*}_j & = & \left[ 
	P^{l*}_j - \gamma^l_j \left( \beta_j P^{l*}_j - \omega^*_j - \lambda_j^* \right)
	\right]_{\underline P^l_j}^{\overline P^l_j}
	\eqn
	for any $\gamma^g_j>0$ and $\gamma^l_j>0$.
	This is equivalent to \eqref{eq:lemma.2}.
	Since $\mu^*=\omega^*=0$, \eqref{eq:lemma.2} is equivalent to
	\eqref{eq:lemma.1} which is equivalent to \eqref{eq:lemma.0}.
	This proves that $(x^*,\rho^*)$ is primal-dual optimal and completes the 
	proof of Theorem \ref{thm:2}.
\end{proof}

Next we prove Theorem \ref{thm:1}.
\begin{proof}[Proof of Theorem \ref{thm:1}] 
	Let $(x^*, \rho^*) = (\theta^*, \omega^*, P^{g*}, P^{l*},$ $\lambda^*, \mu^*)$
	be primal-dual optimal.  
	Lemma \ref{lemma:1} implies that $\omega^*=0$ and $\theta^*$ is unique
	up to the reference angle $\theta_0^*$.     This proves parts 3 and 4 of Theorem \ref{thm:1}.
	
	To prove part 1 of the theorem, 
	since the objective function is strictly convex in $\left(P^g, P^l \right)$ the optimal values
	$\left(P^{g*}, P^{l*}\right)$ are unique.      Hence $x^*$ is unique (up to $\theta_0^*$).   
	
	As for the part 2 of the theorem, from \eqref{eq:lemma.1c} in the proof of Lemma \ref{lemma:2},
	we have  $\mu^*=\omega^* = 0$, implying the uniqueness of $\mu^*$. 
	From \eqref{eq:lemma.1a}\eqref{eq:lemma.1b}, $\lambda_j^*$ is unique if
	either $\underline{P}^g_j<P^{g*}_j<\overline{P}^g_j$ or $\underline{P}^l_j<P^{l*}_j<\overline{P}^l_j$.
	We now prove that this is indeed the case by showing that
	the other four cases cannot hold: (i) $P^{g*}_j = \underline{P}^g_j$ and $P^{l*}_j = \underline{P}^l_j$; 
	(ii) $P^{g*}_j = \underline{P}^g_j$ and $P^{l*}_j = \overline{P}^l_j$; (iii) $P^{g*}_j = \overline{P}^g_j$ and $P^{l*}_j = \underline{P}^l_j$; and (iv) $P^{g*}_j = \overline{P}^g_j$ and $P^{l*}_j = \overline{P}^l_j$. 		
	
	Since $P^{g*}_j - P^{l*}_j=p_j$ for per-node power balance, (ii) and (iii) cannot hold 
	since the inequalities in A2 are strict. 
	Suppose (i) holds.  Then there exists an $\epsilon_j>0$ such
	that $\hat P^{g}_j = \underline{P}^g_j+\epsilon_j<0$ and $\hat P^{l}_j = \underline{P}^l_j+\epsilon_j<0$,
	together with other components of $x^*$, remain a feasible primal solution. 
	However this new feasible solution attains a strictly smaller objective value,
	contradicting the optimality of $x^*$. Thus (i) cannot hold. Similarly (iv) cannot hold. 
	This proves that $\lambda^*$ is unique.
	
	Finally, if $\underline{P}^g_j<P^{g*}_j<\overline{P}^g_j$
	then $\lambda^*_j$ is uniquely determined by $\lambda^*_j = - \alpha_j P^{g*}_j$
	according to \eqref{eq:lemma.1a}.   If 
	$\underline{P}^l_j<P^{l*}_j<\overline{P}^l_j$ then 
	$\lambda^*_j$ is uniquely determined by $\lambda^*_j = \beta P^{l*}_j$
	according to \eqref{eq:lemma.1b}. 
	
	This completes the proof of Theorem \ref{thm:1}.
	
\end{proof}

\section{Proofs of  Lemma \ref{lemma:P(t)} and Theorem \ref{thm:stability.2}}
\renewcommand{\theequation}{B.\arabic{equation}}
\renewcommand{\thetheorem}{B.\arabic{theorem}}
\label{proof:thmstability2}

We prove Lemma \ref{lemma:P(t)} using the first-order inertia dynamics
of \eqref{eq:control.1a} and \eqref{eq:control.1b}.

\vspace{0.1in}
\noindent
\emph{Proof of Lemma \ref{lemma:P(t)}.}
Set 
\bqn
\hat u_j^g(t) & = & \left[ P^g_j(t) - \gamma^g_j \left( \alpha_j P^g_j(t) + \omega_j(t) + \lambda_j(t) \right)
\right]_{\underline P^g_j}^{\overline P^g_j}
\eqn
Then (\ref{eq:model.1c}) can be rewritten as 
\begin{align}
	\label{rewritten}
	T_j^g\dot P_j^g(t) + P_j^g(t) = \hat u_j^g(t)
\end{align}
Apply the Laplace transform to (\ref{rewritten}) to obtain

$${\cal L} \left( P_j^g\right) (s)= {\cal L}\left( \hat u_j^g \right) (s)/(T_j^g s+1).$$

In the time domain $P_j^g (t)$  is then given by convolution:
\begin{align*}
	P_j^g\left( t \right) &=\  \frac{1}{{{T_j^g}}}\int_{{0^ - }}^{ + \infty } {\hat u_j^g\left( {t - \tau } \right){e^{ - \tau /{T_j^g}}} d\tau }  
\\
	&=\  \int_0^{\frac{t}{T^g_j}} \hat u_j^g\left( t - T^g_j \tau \right){e^{ - \tau }}d\tau
\end{align*}
Since $e^{-\tau} >0$ we can replace $\hat u_j^g$ in the integrand
by its lower and upper bounds $\underline P^g_j$ and $\overline P^g_j$ respectively to conclude
\bqn
\int_0^{\frac{t}{T^g_j}} \underline P^g_j \cdot e^{ - \tau } d\tau 
& \le \ \ P_j^g\left( t \right) \ \ \le &
\int_0^{\frac{t}{T^g_j}} \overline P^g_j \cdot e^{ - \tau} d\tau 
\eqn
Hence
\begin{align*}
\underline P_j^g \left( {1 - {e^{ - t/{T_j^g}}}} \right) \ \ \le \ \ P_j^g\left( t \right) 
\ \ \le \ \ \overline P_j^g\left( {1 - {e^{ - t/{T_j^g}}}} \right) 
\end{align*}
and $P_j^g \le P_j^g\left( t \right) \le \overline P_j^g$ for all $t\geq 0$ under
assumptions A1 and A3.
That $\underline P_j^l \le P_j^l\left( t \right) \le \overline P_j^l$ can be proved similarly.
\qed

\vspace{0.1in}
\noindent
\emph{Proof of Theorem \ref{thm:stability.2}.}
We start with a lemma.
\begin{lemma}
	\label{lemma:decreasing.1}
	Suppose A1, A2 and A3 hold. Given any $w(0)\in S$ we have
	\begin{enumerate}
		\item $\dot V_1(w(t))\leq 0, \forall t>0$.
		\item The trajectory $w(t)$ is bounded, i.e., there exists $\overline w$ such
			that $\|w(t)\|\leq \overline w$ for all $t\geq 0$.
	\end{enumerate}
\end{lemma}
	\emph{Proof of Lemma \ref{lemma:decreasing.1}}  
	We omit $t$ in the proof for simplicity. 
	According  \cite[Theorem 3.2]{Fukushima:Equivalent},  since $F(w)$ is continuously 
	differentiable, $V_1(w)$ defined by \eqref{eq:lyapunov.1} is also continuously 
	differentiable. Moreover its gradient is given by
	\bqn
	\nabla_{w}V_1(w)& \!\!\!\!\!\! = & \!\!\!\!\!\! F(w)-(\nabla_wF(w)-I)(H(w)-w)
	+k\Gamma_1^{-2}(w-w^*)
	\eqn 
	Then the derivative of $V_1(w)$ along the solution trajectory is 
	\begin{subequations}
	\begin{align}
	\dot V_1(w) & =  \nabla_{w}^T V_1(w)\cdot \dot w \ \ = \ \ 
				 \nabla_{w}^T V_1(w)\cdot \Gamma_1 (H(w)-w)
	\nonumber \\
		& = (F(w)-\left(\nabla_wF(w)-I)(H(w)-w))^T\Gamma_1 (H(w)-w\right) \nonumber\\
		&\quad +k(w-w^*)\cdot \Gamma_1^{-1} (H(w)-w)
	\nonumber\\
		& = -(H(w)-w)^T\nabla_wF(w)\Gamma_1(H(w)-w) \nonumber\\
		&\quad - \left( H(w) - (w - F(w)) \right)^T\Gamma_1 (w-H(w))  \nonumber\\
		&\quad +k(w-w^*)^T\cdot \Gamma_1^{-1} (H(w)-F(w)+F(w)-w)
		\nonumber\\
		& = -(H(w)-w)^T\nabla_wF(w)\Gamma_1(H(w)-w) \nonumber\\
		&\quad - \left( H(w) - (w - F(w)) \right)^T\Gamma_1 (w-H(w))  \nonumber\\
		&\quad -k(w-w^*)^T\cdot\Gamma_1^{-1} F(w)
	\nonumber\\
		&\quad +k(w-H(w)+H(w)-w^*)^T \Gamma_1^{-1} (H(w)-(w-F(w)))
	\nonumber\\
		& = k(H(w)-w^*)^T\cdot \Gamma_1^{-1} (H(w)-(w-F(w)))	\label{eq:ProofThm.1a}	\\
		&\quad - \left( H(w) - (w - F(w)) \right)^T(\Gamma_1-k\Gamma_1^{-1}) (w-H(w))  \label{eq:ProofThm.1c}\\
		&\quad -(H(w)-w)^T\nabla_wF(w)\cdot\Gamma_1(H(w)-w)\label{eq:ProofThm.1b}\\
		&\quad -k(w-w^*)^T\cdot \Gamma_1^{-1} F(w)	\label{eq:ProofThm.1d}
	\end{align}
	where $\Gamma_1 := \text{diag}\left( 
	B^{-1/2}, M^{-1/2}, (T^g)^{-1}, (T^l)^{-1}, (\Gamma_1^\lambda)^{1/2} \right)$
	is diagonal and positive definite.	
	We now prove that all terms on the right-hand side are nonpositive and
	hence $\dot V_1(w)\leq 0$.
	\label{eq:ProofThm.1}
	\end{subequations}

	\newcounter{TempEqCnt2}
	\setcounter{TempEqCnt2}{\value{equation}}
	\setcounter{equation}{6}
	\newcounter{mytempeqncnt2}
	
		\begin{figure*}[ht]

		\bq
		\nabla_w F(w) & :=& 
		\underbrace{
			\begin{bmatrix}
				B^{-1/2} & 0 & 0 & 0 & 0 \\
				0 & M^{-1/2} & 0 & 0 & 0 \\
				0 & 0 & (T^g)^{-1} & 0 & 0   \\
				0 & 0 & 0 & (T^l)^{-1} & 0   \\
				0 & 0 & 0 & 0 & (\Gamma^\lambda)^{1/2}
			\end{bmatrix}
		}_{\Gamma_1}
		\cdot
		\left[ 
		\begin{array}{c c c c c}
			0         &  -BC^T   &  0     & 0    & 0  \\
			CB      & D           & -I      & I     & 0  \\
			0         & I             & A^g  &0     & I \\
			0         & -I            & 0      & A^l   & -I \\
			0         & 0            & -I      &I      &0\\
		\end{array} 
		\right]
		\label{eq:ProofThm.3}
		\eq
		
		\hrulefill
		\vspace*{4pt}
	\end{figure*}
	\setcounter{equation}{\value{TempEqCnt2}}

	For the term in \eqref{eq:ProofThm.1a} denote the projection of any $w$ 
	onto $S$ under the norm defined by a (symmetric) positive definite matrix $\Gamma$ by
	\bqn
	\text{Proj}_{S, \Gamma}(w) & := & \arg\min_{y\in S}\ ( y-w )^T\Gamma (y-w)
	\eqn
	By the projection theorem a vector $\hat w_\Gamma$ is equal to the projection
	$\text{Proj}_{S, \Gamma}(w)$ if and only if
	\bq
	\left( \hat w_\Gamma - w \right)^T\Gamma \left(y - \hat w_\Gamma \right) 
	& \geq & 0, \qquad y\in S
	\label{eq:ProofThm.2}
	\eq
	Note that $S =: \prod_i S_i$ is a direct product of intervals $S_i$ and
	$\Gamma = \text{diag}(\Gamma_{ii})$ is diagonal.  Hence the projection
	under the $\Gamma$-norm coincide with the projection under the Euclidean
	norm:
	\bqn
	\text{Proj}_{S, \Gamma}(w) & = & 
			\arg\min_{y: y_i\in S_i}\ \sum_i\, \Gamma_{ii} (y_i-w_i)^2
	\\ & = & 
			\arg\min_{y: y_i\in S_i}\ \sum_i\,  (y_i-w_i)^2 \ \ = \ \ \text{Proj}_S(w)
	\eqn
	Substituting into \eqref{eq:ProofThm.2} we have, for any diagonal positive definite $\Gamma$,
	\bq
	\left( \text{Proj}_S(w) - w \right)^T\Gamma \left(y - \text{Proj}_S(w) \right) 
	& \geq & 0, \quad y\in S
	\label{projection property}
	\eq
	for any $w$.  
	The projection $H(w):=\text{Proj}_S(w-F(w))$ of $w-F(w)$ 
	therefore satisfies (for $\Gamma := \Gamma_1^{-1}$)
	\bq
	\label{eq:projineq}
	\left( H(w) - (w-F(w) \right)^T\Gamma_1^{-1} \left(w^* - H(w) \right) 
	& \geq  & 0
	\eq
	since $w^*\in S$.  This proves
	 that the right-hand side of \eqref{eq:ProofThm.1a} is nonpositive

	To show that (\ref{eq:ProofThm.1c}) is nonpositive we use a similar argument.
	Since $\Gamma := \Gamma_1-k\Gamma_1^{-1} >0$ we can define 
	the projection $\text{Proj}_{S, \Gamma}(w)$ under this $\Gamma$.	 As
	explained above $\text{Proj}_{S, \Gamma}(w) = \text{Proj}_{S}(w)$ and
	hence as before, we have
	\bqn
	\left( H(w) - (w-F(w) \right)^T\Gamma \left(w - H(w) \right) 
	 & \geq  & 0
	\eqn
	since the solution trajectory $w(t)\in S$ for all $t\geq 0$ by Lemma 
	\ref{lemma:P(t)}.  This proves the term in \eqref{eq:ProofThm.1c} is nonpositive.
	
	We will prove that (\ref{eq:ProofThm.1d}) is nonpositive. 
	Along any solution trajectory we always have $\mu(t)\equiv \omega(t)$.
	Substituting into the Lagrangian $L_1(x, \rho)$ in \eqref{eq:defL.1} we
	obtain a function 
	\bqn
	& & \hat L_1(\tilde\theta, P^g, P^l, \lambda, \omega)   
	\ \ := \ \  
	L_1(\theta, \omega, P^g, P^l, \lambda, \omega)
	\\ & = & 
	\frac{1}{2} \left( (P^g)^T A^g P^g + (P^l)^T A^l P^l - \omega^T D \omega \right)
	\\
	& & + \ \lambda^T \!\! \left( P^g - P^l - p \right) + \ \omega^T \!\! \left( P^g - P^l - p - C B\tilde \theta \right)
	\eqn
Write $w_1 := (\tilde\theta, P^g, P^l)$, $w_2 := ( \lambda, \omega)$.
Then $\hat L_1(w_1, w_2)$ is convex in $w_1$ and concave in $w_2$.
It can be verified that \footnote{For notational
simplicity, we have re-arranged the order of the variables in $w$ to 
$w := (\tilde\theta, P^g, P^l, \lambda, \omega)$ and components of $F$ to match the
order of $(w_1, w_2)$.}
\bqn
\Gamma_1^{-1} F(w) & \!\!\!\! = \!\!\!\! & \begin{bmatrix}
	\nabla_{\tilde\theta} \hat L_1 \\
	\nabla_{P^g} \hat L_1 \\
	\nabla_{P^l} \hat L_1 \\
	- \nabla_{\lambda} \hat L_1 \\
	- \nabla_{\omega} \hat L_1
	\end{bmatrix} (w_1, w_2)
	\ = \ 
	\begin{bmatrix} \ \ \nabla_{w_1}\hat L_1 \\ - \nabla_{w_2} \hat L_1 \end{bmatrix}\!
	(w_1, w_2)
\eqn
	Hence, we have
	\begin{align}
		&\ -k(w-w^*)^T\cdot \Gamma_1^{-1} F(w) \nonumber \\
		= &\  -k(w_1-w_1^*)^T \nabla_{w_1}\hat L_1(w_1,w_2) + k(w_2-w_2^*)^T \nabla_{w_2}\hat L_1(w_1,w_2) \nonumber \\
		\le & \ k \bigg(\hat L_1(w_1^*,w_2) - \hat L_1(w_1,w_2) +\hat L_1(w_1,w_2) - \hat L_1(w_1,w^*_2)\bigg) \nonumber \\
		= &  \ k\bigg(\hat L_1(w_1^*,w_2) - \hat L_1(w^*_1,w^*_2) + \hat L_1(w^*_1,w^*_2) - \hat L_1(w_1,w^*_2)\bigg)\nonumber \\
		\le & \ 0	
	\end{align}
   where the first inequality follows because $\hat L_1$ is convex in $w_1$ and concave in $w_2$ and the second inequality follows
   because $(w^*_1, w^*_2)$ is a saddle point.    
   Therefore (\ref{eq:ProofThm.1d}) is nonpositive.  
   
Finally to prove that \eqref{eq:ProofThm.1b} is nonpositive note that
	\bqn
	(H(w)-w)^T\nabla_wF(w)\Gamma_1(H(w)-w) & \!\!\!\! = \!\!\!\! & 
			\dot w^T \left( \Gamma_1^{-1} \nabla_wF(w) \right) \dot w
	\eqn
	where $\dot w := (\dot{\tilde\theta}, \dot\omega, \dot P^g, \dot P^l, \dot\lambda)$.
	From \eqref{eq:Fz.1}, $\nabla_w F(w)$ is given by \eqref{eq:ProofThm.3}.

	Hence  
	\bqn
	\dot w^T \left(\Gamma_1^{-1} \nabla_w F(w) \right) \dot w
	& \!\!\!\! = \!\!\!\! &  \dot\omega^T D \dot\omega \, + \, \dot P^{gT} A^g \dot P^{g} \, +\, \dot P^{lT} A^l \dot P^l 
	\, \geq \, 0
	\eqn
	and hence \eqref{eq:ProofThm.1b} is nonpositive.
   
 This also implies that  
 	\setcounter{equation}{7}
	\begin{align}
	\dot V_1(w(t)) & \leq  - \left(
		\dot\omega^T D \dot\omega \, + \, \dot P^{gT} A^g \dot P^{g} \, +\, \dot P^{lT} A^l \dot P^l \right) 
		\ \, \leq \ \, 0
		\label{eq:ProofThm.4}
	\end{align}
for all $t\geq 0$.  This proves the first assertion of the lemma.

To prove that the trajectory $w(t)$ is bounded
   note that \cite[Theorem 3.1]{Fukushima:Equivalent} proves that 
   $\hat V_1(w) := - \left( H(w)-w \right)^T F(w)  \, - \, \frac{1}{2} ||H(w)-w||^2_2$ satisfies
$\hat V_1(w) \geq 0$ over $S$.  Hence
\bqn
\frac{1}{2}k(w(t)-w^*)^T\Gamma_1^{-2}(w(t)-w^*) &\!\!\! \le \!\!\! & V_1(w(t)) \ \le \ V_1(w(0))
\eqn
   indicating the trajectory $w(t)$ is bounded, as desired.
\qed

\begin{lemma}
	\label{lemma:LaSalle}
	Suppose A1, A2 and A3 hold. Given any $w(0)\in S$, we have
	\begin{enumerate}
		\item The trajectory $w(t)$  converges to the largest invariant set 
			$W_1^*$ contained in $W_1=\{w\in S| \ \dot P^g= \dot P^l= \dot \omega=0\}$.
		\item Every point $w^*\in W_1^*$ is an equilibrium point of \eqref{eq:model.3}.
	\end{enumerate}
\end{lemma}
\begin{proof}[Proof of Lemma \ref{lemma:LaSalle}]
Fix any initial state $w(0)$ and consider the trajectory $(w(t), t\geq 0)$ of the
closed-loop system \eqref{eq:model.3}.
Lemma \ref{lemma:decreasing.1} implies a compact set 
$\Omega_0 := \Omega(w(0)) \subset S$ such that $w(t)\in\Omega_0$ 
for $t\geq 0$ and $\dot V_1(w) \leq 0$ in $\Omega_0$.
Let $W_1 :=\{w\in \Omega_0 | \ \dot P^g= \dot P^l= \dot \omega=0\}$.
Then \eqref{eq:ProofThm.4} implies that $w\in W_1$ if and only if
$\dot V_1(w) = 0$.
According to LaSalle's invariance principle (\cite[Theorem 4.4]{Khalil:Nonlinear})
the solution trajectory $(w(t), t\geq 0)$ converges to the largest invariant set 
contained in $W_1$, proving the first assertion.

For the second assertion, fix any $w(0)\in W_1^*$.  We claim that $w(0)$ 
must be an equilibrium point of \eqref{eq:model.3}.
Since $W_1^*$ is invariant we have
\bq
\dot P^g(t) \ = \ \dot P^l(t) \ = \ \dot \omega(t) \ = \ 0, \quad t\geq 0
\label{eq:ProofThm.6}
\eq
It suffices to prove that $\dot w(t) = 0$ for $t\geq 0$, i.e., 
$\dot{\tilde\theta} = 0$ and $\dot\lambda = 0$ for $t\geq 0$.

Since $P^g(t), P^l(t), \omega(t)$ are bounded (Lemma \ref{lemma:decreasing.1}),
\eqref{eq:ProofThm.6} implies that
$$(P^g(t), P^l(t), \omega(t)) \ \equiv \ (P^{g\infty}, P^{l\infty}, \omega^{\infty}) $$
for some finite constants $(P^{g\infty}, P^{l\infty}, \omega^{\infty})$.
Hence 
\bqn
\dot{\tilde\theta}(t) & = & C^T \omega^\infty \ \ = \ \ \text{constant}
\eqn
implying that $\tilde\theta(t)$ grows linearly in $t$, contradicting that $\tilde\theta(t)$ is bounded
unless $\dot{\tilde\theta} = 0$ for $t\geq 0$.   Similarly
	\bqn
	\dot{\lambda}(t) & = & \Gamma^{\lambda}\left(P^{g\infty} -P^{l\infty}-p\right)
	\ \ = \ \ \text{constant}
	\eqn
Hence the boundedness of $\lambda(t)$ implies that $\dot{\lambda}(t)= 0$ for $t\geq 0$.
This proves that any $w(0) \in W_1^*$ is an equilibrium point.
\end{proof}
 
If all inequalities in A2 are strict, then the equilibrium point $w^*$ of the closed-loop
system \eqref{eq:model.3}
is unique (Theorem \ref{thm:1}.2) and Lemma \ref{lemma:LaSalle} implies that $w(t)$
converges to $w^*$ \cite[Corollary 4.1, p. 128]{Khalil:Nonlinear} as $t\to \infty$. 
When there are multiple equilibrium points, Lemma \ref{lemma:LaSalle} 
is not adequate to conclude asymptotic stability.
We use instead a more direct argument due to \cite{Changhong:Design, LiZhaoChen2016}.

\begin{proof}[Proof of Theorem \ref{thm:stability.2}]
Fix any initial state $w(0)$ and consider the trajectory $(w(t), t\geq 0)$ of the
closed-loop system \eqref{eq:model.3}.
As mentioned in the proof of Lemma \ref{lemma:LaSalle},
$w(t)$ stays entirely in a compact set $\Omega_0$.   Hence there exists an infinite 
sequence of time instants ${t_k}$ such that $w(t_k)\to \hat {w}^*$ as $k\to\infty$, 
for some $\hat w^*$ in $W_1^*$.  Lemma \ref{lemma:LaSalle} guarantees that
$\hat w^*$ is an equilibrium point of the closed-loop system \eqref{eq:model.3}
and hence $H(\hat w^*) = \hat w^*$.
Use this specific equilibrium point $\hat {w}^*$ in the definition of $V_1$ 
in (\ref{eq:lyapunov.1}) to get the Lyapunov function:
\begin{align*}
	V_1(w) & =  - \left( H(w)-w \right)^T F(w)  \, - \, \frac{1}{2} ||H(w)-w||^2_2 \\
	&\quad +\frac{1}{2}k(w-\hat w^*)^T\Gamma_1^{-2}(w-\hat w^*) 
\end{align*}
Since $\dot V_1\leq 0$, $V_1(w(t))$ converges.  Moreover
it follows from the continuity of $V_1$ that 
\bqn
\lim_{t\to\infty} V_1(w(t)) \ \ = \ \ \lim_{k\to\infty} V_1(w(t_k)) & = & V_1(\hat w^*) \ \ = \ \ 0
\eqn 
The quadratic term $(w-\hat w^*)^T\Gamma_1^{-2}(w-\hat w^*)$
in $V_1$ then implies that $w(t)\to \hat {w}^*$ as $t\to \infty$.
\end{proof}